\documentclass{article}

\usepackage{microtype}
\usepackage{graphicx}
\usepackage{booktabs} % for professional tables
\usepackage{subcaption}
\usepackage{bm}

\usepackage{hyperref}

\usepackage[accepted]{icml2024}

\usepackage{amsmath}
\usepackage{amssymb}
\usepackage{mathtools}
\usepackage{amsthm}

\usepackage{listings}
\usepackage{xcolor}

\usepackage{fancyhdr}
\usepackage[capitalize,noabbrev]{cleveref}

\theoremstyle{plain}
\newtheorem{theorem}{Theorem}[section]
\newtheorem{proposition}[theorem]{Proposition}
\newtheorem{lemma}[theorem]{Lemma}

\theoremstyle{definition}
\newtheorem{definition}[theorem]{Definition}
\newtheorem{assumption}[theorem]{Assumption}
\theoremstyle{remark}
\newtheorem{remark}[theorem]{Remark}

\newcommand\inner[2]{\langle #1, #2 \rangle}

\newcommand{\E}{\mathbb{E}}

\newcommand{\R}{\mathbb{R}}
\newcommand{\p}{P}
\newcommand{\ptrue}{\mathbb{\p}^*}
\newcommand{\grad}{\nabla}

\newcommand{\ourapproach}{\textrm{MC-EIF}}

\newcommand\inv[1]{#1^{\text{-}1}}

\usepackage [english]{babel}
\usepackage [autostyle, english = american]{csquotes}
\MakeOuterQuote{"}

\usepackage[textsize=tiny]{todonotes}

\icmltitlerunning{Automated Efficient Estimation using Monte Carlo Efficient Influence Functions}

\begin{document}

\twocolumn[
\icmltitle{Automated Efficient Estimation using\\ Monte Carlo Efficient Influence Functions}



\icmlsetsymbol{equal}{*}

\begin{icmlauthorlist}
\icmlauthor{Raj Agrawal}{basis,broad}
\icmlauthor{Sam Witty}{basis,broad}
\icmlauthor{Andy Zane}{basis,umass}
\icmlauthor{Eli Bingham}{basis,broad}
\end{icmlauthorlist}

\icmlaffiliation{umass}{College of Information and Computer Sciences, UMass, Amherst, United States}
\icmlaffiliation{basis}{Basis Research Institute, NYC, United States}
\icmlaffiliation{broad}{Broad Institute of MIT and Harvard, Cambridge, United States}

\icmlcorrespondingauthor{Raj Agrawal}{raj@basis.ai.}

\icmlkeywords{Machine Learning, Efficient Influence Function}

\vskip 0.3in
]



\printAffiliationsAndNotice{}  % leave blank if no need to mention equal contribution

\begin{abstract}
Many practical problems involve estimating low dimensional statistical quantities with high-dimensional models and datasets. Several approaches address these estimation tasks based on the theory of influence functions, such as debiased/double ML or targeted minimum loss estimation. This paper introduces \textit{Monte Carlo Efficient Influence Functions} (MC-EIF), a fully automated technique for approximating efficient influence functions that integrates seamlessly with existing differentiable probabilistic programming systems. MC-EIF automates efficient statistical estimation for a broad class of models and target functionals that would previously require rigorous custom analysis. We prove that MC-EIF is consistent, and that estimators using MC-EIF achieve optimal $\sqrt{N}$ convergence rates. We show empirically that estimators using MC-EIF are at parity with estimators using analytic EIFs. Finally, we demonstrate a novel capstone example using MC-EIF for optimal portfolio selection.
% \let\thefootnote\relax\footnotetext{\tiny{Distribution Statement A: Approved for public release. Distribution is Unlimited.}}
 %\footnote{ \tiny{\noindent Distribution Statement A: Approved for public release. Distribution is Unlimited.}}
\end{abstract}

\section{Introduction}
\label{sec:introduction}

Over the past several decades, there has been remarkable research progress in methods for robust and efficient statistical estimation, especially for high dimensional problems. A particularly compelling class of such methods are built on a foundation of \textit{efficient influence functions} (EIF), i.e., functional derivatives in the space of probability distributions~\cite{semiparametric_review}. These methods have been particularly fruitful in causal inference applications, where estimating quantities such as the average treatment effect require modeling high-dimensional nuisance parameters relating confounders to treatment and outcome variables. Intuitively, these methods focus finite statistical resources on quantities that matter, and not on  nuisance parameters that only indirectly inform the statistical quantities we wish to estimate.

Despite their successes, estimation methods based on the EIF have lagged behind the kinds of automation that machine learning practitioners have grown accustomed to, instead requiring complex manual derivation on a case-by-case basis. This is contrasted with the generality of automatic differentiation (AD)~\cite{ad_survey} systems and probabilistic programming languages (PPLs) such as Pyro~\cite{pyro}, which automate numerical computations for probabilistic inference. EIF-based estimators have historically eluded this level of automation and generality, because exact recovery of the EIF requires solving high-dimensional integral equations.

\textbf{Contributions.} In this paper we introduce \textit{Monte Carlo Efficient Influence Functions} (MC-EIF), a general and automated technique for numerically computing EIFs using only quantities that are already available from existing AD and PPL systems. Our key insight is that EIFs can be expressed equivalently as a product of (i) the gradient of the functional, (ii) the inverse Fisher information matrix, and (iii) the gradient of the log-likelihood, as shown in \cref{thm:eif_fn} in \cref{sec:MC-EIF}. In \cref{sec:programmable-DRI}, we show how MC-EIF can be used to automatically construct a variety of efficient estimators for a broad class of models and functionals, avoiding the need for complex manual and error-prone derivations.

In summary, we show that: (i) MC-EIF provides accurate estimates of the true EIF, enabling efficient estimation, and (ii) MC-EIF is very general, applying to many functionals and models that can be written as probabilistic programs\footnote{We make these conditions precise in \cref{sec:MC-EIF}.}.

% applying to any functional which is differentiable almost everywhere and to any model with a likelihood that is differentiable almost everywhere and an invertible Fisher information matrix.

\textbf{Outline.} We review related work in \cref{sec:related-work} and provide background in \cref{sec:problem-statement}. In \cref{sec:MC-EIF}, we introduce MC-EIF and provide a non-asymptotic error bound on the quality of our approximation. We show how estimators using MC-EIF achieve the same asymptotic guarantees as using analytic EIFs in \cref{sec:programmable-DRI}. In \cref{sec:experiments}, we show empirically that MC-EIF produces more accurate estimates of the EIF than existing automated approaches, and using MC-EIF as a drop-in replacement for the analytic EIF does not degrade estimation accuracy in a variety of benchmarks.

% In \cref{sec:MC-EIF}, we prove that MC-EIF is asymptotically consistent and that estimators using MC-EIF converge at the same root-n rate as those using the exact EIF. In \cref{sec:experiments} we show empirically that MC-EIF produces more accurate estimates of the EIF than existing automated approaches and using MC-EIF as a drop-in replacement for the analytic EIF produces only negligible reductions in estimation accuracy in a variety of causal inference benchmarks. Finally, we show that MC-EIF expands the reach of existing EIF-based estimators by presenting a EIF-based estimator on a optimal decision making task with a kernel ridge regression model.

\subsection{Related Work}
\label{sec:related-work}

Influence function-based estimators have a rich history in the statistics and machine learning literature~\cite{efficient_book,  semi-theory-book, semiparametric_review, review_demystifying}. Despite their effectiveness, these methods have historically required custom and complex mathematical analysis for specific combinations of models and functional. Even an incomplete survey of the influence function-based estimation literature yields a large collection of complex scenario-specific research. We enumerate some of this research in~\cref{A:Survey}. Importantly, our work does not aim to introduce novel efficient estimators; instead it aims to lower the mathematical burden for practitioners who wish to use existing influence function-based efficient estimator templates (see \cref{sec:programmable-DRI}) with custom models and/or functionals.

Our work is not the first to attempt to automate and generalize computations for robust statistical estimation. Perhaps the closest technique we are aware of is approximating the influence function using finite differences on kernel-smoothed empirical distributions~\cite{automated_if, computer-eif, empirical-gateaux}. We provide a thorough comparison with this method in \cref{sec:experiments}, demonstrating how MC-EIF automates and scales better to high dimensional problems, exactly the settings where efficient estimation is most useful. Recent work has made progress towards general efficient estimators, but still impose strong restrictions on models and/or functionals. DML-ID~\cite{double-graph} extends double machine learning to nonparametric causal graphs and marginal density under intervention functionals. Similarly, the kernel debiased plugin estimator~\cite{kernel-debiased} implements a version of TMLE that bypasses the influence function computations for models defined in a RKHS. Finally, other recent work \cite{auto-dml,orth_stat_learn, double-nn, double-rf} approximates the efficient influence function using a meta-regression procedure for generalized method-of-moment estimators.

Our work on influence functions for estimation should not be confused with other recent work~\cite{empirical-influence, infinitesimal-jackknife} automating numerical computations on what is often called the \textit{empirical influence function}. The empirical influence function targets the effect of perturbing individual training points in a machine learning model, and not the effect of perturbing an estimated distribution as the EIF does. While we take inspiration from this prior work, and borrow similar techniques, it is not directly applicable to the robust estimation tasks we take on in this paper.

\section{Problem Statement}
\label{sec:problem-statement}

We begin with an overview of semiparametric statistics, highlighting its role in developing optimal estimators via influence functions. Then, we address the practical challenges in constructing general semiparametric estimators.

\subsection{Overview of Semiparametric Statistics}

% \textbf{Motivation.} Consider the following motivating example from causal inference: an analyst would like to estimate the average treatment effect (ATE) of a binary intervention using $N$ observational datapoints $\{(c_n, t_n, y_n)\}_{n=1}^N$, where $c_n$, $t_n$, and $y_n$ denote the observed confounders, treatment assignment, and outcome for individual $n$, respectively. To estimate the ATE, an analyst might regress $Y$ on both the treatment assignment and confounder. The estimated ATE is then the average of the predicted differences: $\frac{1}{N}\sum_{n=1}^N [\hat{m}(c_n, 1) - \hat{m}(c_n, 0)]$, where $\hat{m}$ denotes the fitted regression function targeting $m(c, t) = \E[y \mid c=c, t=t]$. Assuming correct model specification and sufficient data, this \emph{regression-adjustment} approach tends to yield a consistent estimate of the ATE as $N \rightarrow \infty$. Unfortunately, if $c$ is high-dimensional or $m(\cdot, \cdot)$ is not sufficiently smooth, $\hat{m}$ will have high variance, leading to a poor estimate of the ATE~\cite{dml}.

\textbf{General Problem.} We consider the problem of estimating some estimand $\theta^* \in \R^L$, where $L \in \mathbb{N}$ denotes the dimension of the target quantity. Typically, we can express $\theta^* = \Psi(\ptrue)$ for some known functional $\Psi$, where $\Psi$ maps a probability distribution to a vector in $\R^L$, and $\ptrue(x)$ denotes the true data-generating distribution over some vector of observables $x \in \R^D$, $D \in \mathbb{N}$. Many estimands involve estimating high-dimensional \emph{nuisance} parameters, or quantities of no immediate value to the analyst. For example, when estimating the average treatment effect, one might need to adjust for high-dimensional confounders.

% For example, in the ATE example, $x = (c, t, y)$ and $\Psi(\mathbb{P})$ equals
%
% \begin{equation}
% 	\E_{c \sim \mathbb{P}(c)}\left[\E_{y \sim \mathbb{P}(y \mid c, t=1)}[y] - \E_{y \sim \mathbb{P}(y \mid c, t=0)}[y]\right].
% \end{equation}

\textbf{Semiparametric Solution.} Semiparametric statistics provides a mathematical framework for optimally estimating $\theta^*$ in the presence of potentially complex, high-dimensional nuisance parameters.
A standard way to estimate $\theta^*$ is with the \emph{plug-in approach}; construct an estimate $\hat{\mathbb{P}}$ of $\ptrue$ and report $\hat{\theta} = \psi(\hat{\mathbb{P}})$. Unfortunately, the plug-in approach can lead to provably sub-optimal estimates of $\theta^*$ due to poor estimates of $\mathbb{P}^*$~\cite{semi-theory-book,dml,orth_stat_learn}. Instead, a general recipe for efficiently estimating $\theta^*$ from finite data $\{x_n\}_{n=1}^N$, where $x_n \overset{\text{iid}}{\sim} \mathbb{P}^*(x)$ for $n=1, \cdots, N$, is given by the following three-steps: (i) use $\frac{N}{2}$ samples to construct an initial estimate $\hat{\mathbb{P}}$ of $\ptrue$, (ii) compute the influence function (to be defined shortly) of $\Psi$ at the estimate $\hat{\mathbb{P}}$, and (iii) evaluate the influence function at the held out $\frac{N}{2}$ datapoints to derive a corrected estimate.\footnote{There are some small variations to this three-step recipe such as using cross-fitting~\cite{dml} instead of a simple equal split of the data; see also \cref{sec:programmable-DRI}.} In \cref{sec:programmable-DRI}, we elaborate on how influence functions are used to construct several popular efficient estimators.

\subsection{Influence Functions}

 A central premise of this paper is that to automate efficient estimation, it suffices to automate the computation of \emph{influence functions}, which can be thought of as gradients in function space. We make this precise below.
\begin{definition} \label{def:gateaux} (Gateaux derivative)
Consider the $\epsilon$-perturbed probability distribution $\mathbb{P}_{\epsilon} \coloneqq (1 - \epsilon)\mathbb{P} + \epsilon \mathbb{Q} = \mathbb{P} + \epsilon(\mathbb{Q} - \mathbb{P})$, where $\mathbb{Q}$ is some probability distribution. $\Psi$ is \emph{Gateaux} differentiable at $\mathbb{Q}$ if the following limit exists:
\begin{equation*}
	 \left.\frac{d}{d\epsilon} \Psi(\mathbb{P}_{\epsilon}) \right \vert_{\epsilon = 0} = \lim_{\epsilon \rightarrow 0} \frac{\Psi(\mathbb{P}_{\epsilon}) - \Psi(\mathbb{P})}{\epsilon}.
\end{equation*}
\end{definition}
The Gateaux derivative can be viewed as a generalization of the directional derivative from ordinary differential calculus; it characterizes how much a functional changes at a particular point $\mathbb{P}$ in the direction $\mathbb{Q} - \mathbb{P}$. 
\begin{definition} \label{def:if} (Influence function) Suppose there exists a square integrable function $\varphi \in L^2(\mathbb{P})$ such that
\begin{equation*}
	%\begin{split}
		\left.\frac{d}{d\epsilon} \Psi(\mathbb{P}_{\epsilon}) \right \vert_{\epsilon = 0} = \inner{\varphi}{\mathbb{Q} - \mathbb{P}}_{L^2} 
					= \int_{x \in \R^D} \varphi(x) d(\mathbb{Q} - \mathbb{P})(x)
	%\end{split}
\end{equation*}
for all $\mathbb{Q} \in \mathcal{M}$ and $\E_{x \sim \mathbb{P}(x)}[\varphi(x)] = 0$, where $\mathcal{M}$ denotes some space of probability distributions. Then, $\varphi$ is called an influence function for $\Psi$ at $\mathbb{P}$.

\end{definition}
An influence function is a re-centered "functional gradient" in $L^2(\mathbb{P})$: just as the Euclidean inner product between the gradient of a function and a vector yields the directional derivative in ordinary differential calculus, the $L^2(\mathbb{P})$ inner product between the influence function and perturbation $\mathbb{Q} - \mathbb{P}$ yields the Gateaux directional derivative. Unlike in ordinary differential calculus, however, influence functions are not always unique~\cite{semi-theory-book,kennedy-semi-emp-review}, and some may lead to higher asymptotic variance estimators than other. We call the optimal influence function (i.e., the one that minimizes asymptotic variance), the \emph{efficient influence function} (EIF). When the EIF exists, it is $\mathbb{P}$ almost everywhere unique, and found through a Hilbert space projection onto what is known as the nuisance tangent space. We defer details to \citet{semi-theory-book} and \citet{kennedy-semi-emp-review}. 

\subsection{The Problem: Solving Integral Equations is Hard}

As the influence function in \cref{def:if} is defined implicitly as a solution to an infinite set of integral constraints over $\mathcal{M}$, it is often hard to find. Entire papers have been written to analytically derive influence functions; see, for example, the papers listed in \cref{A:Survey}. For even experts in machine learning and statistics, such derivations are out-of-reach, time consuming, and error prone. 

% The main automation tool so far has been to approximate the influence function via a finite-difference approach~\cite{empirical-gateaux, computer-eif}. In particular, if $\mathbb{Q} = \delta_{x^{\prime}}$, then the Gateaux derivative equals $\varphi({x^{\prime}})$. Unfortunately, setting $\mathbb{Q} = \delta_{x^{\prime}}$ and picking $\epsilon$ small in \cref{def:gateaux} is extremely numerically unstable; see Figure 1 in~\cite{computer-eif}. Specifically, $\epsilon$ needs to decay exponentially with the input dimension $D$ for accurate approximation. Such precision is not implementable with finite floating point precision on a computer when $D \approx 10$~\cite{computer-eif}. 

% Moreoever, for semiparametric models, this is suboptimal because it does not target the efficient influence function.

%\begin{definition}
%Suppose there exists a $\varphi$ that satisfies \cref{def:if} for some Hilbert space probability of probability distributions $\mathcal{H}$. Then, the efficient influence function is the function
%\end{definition}

\section{Monte Carlo Efficient Influence Function}
\label{sec:MC-EIF}

Much of the work in semiparametric statistics and efficient estimation has focused on scenarios where the nuisance function is modeled nonparametrically~\cite{semi-theory-book,tmle, dml, semiparametric_review}. However, practitioners often use high-dimensional parametric models such as generalized linear models, neural networks, and tensor splines in practice due to their flexibility and ability to scale to large datasets. Due to the richness of these high-dimensional spaces, inference is still statistically challenging and  benefits from efficient estimation; see, for example, Table 1 in \citet{dml}. Specifically, in contrast to traditional low-dimensional parametric models where maximum likelihood estimation is typically efficient~\cite{fisher-classic,cramer-bound}, high-dimensional parametric models often exhibit distinct asymptotic behaviors~\cite{,vaart-stats, karoui-high-dim-regression, sparsity-book}. In these high-dimensional models, estimates may converge   slower than classic $O_p(\frac{1}{\sqrt{N}})$ rates without the application of efficient inference methods~\cite{tmle,dml,semiparametric_review}.

A key question we address is whether using a high-dimensional parametric model simplifies the process of solving \cref{def:if}. We show that it does in \cref{sec:theo_eif}.

\subsection{The EIF in Parametric Models} \label{sec:theo_eif}
\textbf{Notation.} We let $\phi \in \Phi \subset \R^p$ denote a finite-dimensional parameter specifying a distribution on the observed random variables $x \in \R^D$ for $p < \infty$, $p \in \mathbb{N}$. $p_{\phi}(x)$ corresponds to a distribution in this space,  and $p_{\phi^*}(x)$ the true distribution, or the one closest to the true data-generating distribution in Kullback–Leibler distance. We let $\psi(\phi)$ denote a function $\R^p \mapsto \R^L$ that equals the evaluation of the functional $\Psi(p_{\phi})$ for all $\phi \in \Phi$. Under mild differentiability assumptions, we provide the analytic formula for the EIF in \cref{thm:eif_fn}.
	
\begin{assumption} \label{assum:cont_diff_prob}
	For any $x \in \R^D$, $p_{\phi}(x)$ is a continuous and differentiable probability density function of $\phi$.\footnote{This assumption does not require $x$ to be continuous, only that the distribution $p_{\phi}(x)$ be continuous in parameter space.}
\end{assumption}

\begin{assumption} \label{assum:cont_diff_func}
	$\psi(\phi)$ is a continuous and differentiable function of $\phi$.
\end{assumption}

\begin{assumption} \label{assum:fisher_inv} The Fisher information $I(\phi)$ is invertible, where
$I(\phi) \coloneqq \E_{x \sim p_{\phi}(x)} [\nabla_{\phi} \log p_{\phi}(x) \nabla_{\phi} \log p_{\phi}(x)^T]$.
\end{assumption}

\begin{theorem} \label{thm:eif_fn} (Theorem 3.5 in \citet{semi-theory-book})
	Suppose  \cref{assum:cont_diff_prob}, \cref{assum:cont_diff_func}, and \cref{assum:fisher_inv} hold. Then, the efficient influence function $\varphi_{\phi}(\tilde{x})$ at $\phi$ evaluated at the point $\tilde{x} \in \mathbb{R}^D$ equals
\begin{equation} \label{eq:eif}
	\begin{split}
   		& [\nabla_{\phi} \psi(\phi)]^T \inv{I(\phi)} \nabla_{\phi} \log p_{\phi}(\tilde{x}).
   \end{split}
\end{equation}
\end{theorem}
While \cref{eq:eif} has been around for many decades, it has mainly been used as a theoretical tool for mathematical statisticians. In particular, \cref{eq:eif} is typically evaluated at the true data generating parameter $\phi^*$ to characterize the theoretical asymptotic variance of an estimator. In other instances, it is used to derive approximate confidence intervals; see, for example, Chapter 3 in \citet{semi-theory-book}. In the following Sections, we discuss how \cref{eq:eif} provides a key ingredient in automating efficient estimation in high-dimensional parametric models.

\subsection{Numerically Approximating the EIF} \label{sec:monte_eif}

Given a model $p_{\phi}(\cdot)$ and functional $\Psi(\cdot)$, we seek to automatically compute \cref{eq:eif}. Our \emph{Monte Carlo efficient influence function} (MC-EIF) estimator achieves this automation by replacing $\psi(\phi)$ and $I(\phi)$, which are typically unknown, with stochastic approximations $\hat{\psi}_M(\phi)$ and $\hat{I}_M(\phi)$ computed from $M \in \mathbb{N}$ Monte Carlo samples:
\begin{equation} \label{eq:mceif}
    \hat{\varphi}_{\phi, M}(\tilde{x}) \coloneqq [\nabla_{\phi} \hat{\psi}_M(\phi)]^T \inv{\hat{I}_M(\phi)} \nabla_{\phi} \log p_{\phi}(\tilde{x}).
\end{equation}
Here, we show that \cref{eq:mceif} leads to an automated and accurate approach to numerically computing EIFs using only quantities provided by existing AD and PPL systems.

\textbf{Approximating \bm{$\inv{\hat{I}_M(\phi)} \nabla_{\phi} \log p_{\phi}(\tilde{x})$}.} We draw $x_m \overset{\text{iid}}{\sim} p_{\phi}(x), 1 \leq m \leq M$ for $M \in \mathbb{N}$, and let
\begin{equation} \label{eq:emp_fisher}
    \begin{split}
    \hat{I}_M(\phi) &= \frac{1}{M} \sum_{m=1}^{M} \nabla_{\phi} \log p_{\phi}(x_m) \nabla_{\phi} \log p_{\phi}(x_m)^T.
    \end{split}
\end{equation}
A naive approach for computing $\inv{\hat{I}_M(\phi)} \nabla_{\phi} \log p_{\phi}(\tilde{x})$ is calculating the full $p \times p$ matrix in \cref{eq:emp_fisher}, inverting it, and then taking its product with the score vector $\nabla_{\phi} \log p_{\phi}(x_m)^T \in \R^p$ computed from AD. This naive approach takes $O(Mp^2 + p^3)$ time and $O(p^2)$ memory which might be too expensive for large $p$. Instead, we exploit AD and numerical linear algebra techniques to avoid explicitly storing and inverting the approximate Fisher information matrix, similar to~\citet{empirical-influence}. Suppose for the moment that we have a black-box method to compute Fisher vector products $\hat{I}_M(\phi) v$ for arbitrary vectors $v \in \R^p$. Then, we could use the conjugate gradient algorithm to iteratively find $\inv{\hat{I}_M(\phi) \nabla_{\phi}} \log p_{\phi}(\tilde{x})$, where the cost of each conjugate gradient step is determined by the cost to compute $\hat{I}_M(\phi) v$. While the number of conjugate gradient steps needs to be $p$ for an exact inverse, often far fewer iterations are required for a close approximate solution~\cite{exact-gp}. To make computing $\hat{I}_M(\phi) v$ efficient, we collect the $M$ simulated datapoints in the matrix $X_M \in \R^{M \times D}$ and let
\begin{equation*}
    \log p_{\phi}(X) \coloneqq (\log p_{\phi}(x_1), \cdots, \log p_{\phi}(x_M))^T \in \R^M.
\end{equation*}
Then, $\hat{I}_M(\phi) \nabla_{\phi} v$ equals
\begin{equation}
   \left[\frac{1}{M} J_M^T J_M\right] v = \left[\frac{1}{M} J_M^T \right] \left[J_M v \right],
\end{equation}
where $J_M = \nabla_{\phi} \log p_{\phi}(X_M) \in \R^{M \times p}$ is the Jacobian matrix. We use \emph{Pearlmutter's trick} to avoid computing the entire Jacobian matrix~\cite{perlmutter-trick}. In particular, this method allows us to compute the Jacobian vector product $v_M = \left[J_M v \right] \in \R^M$ in time proportional to a single evaluation of $\log p_{\phi}(X)$ and $O(M + p)$ memory. Similarly, we use the vector Jacobian product to compute $J_M^Tv_M$.

\textbf{Approximating \bm{$\nabla_{\phi} \hat{\psi}_M(\phi)$}.} Robust estimation with MC-EIF does not require exact gradients. Instead, it only requires access to a sequence of gradient estimators $\{\grad_{\phi} \hat{\psi}_m(\phi) \}_{m=1}^{\infty}$ of $\nabla_{\phi} \psi(\phi)$ whose error can be bounded above by some $\Delta_m > 0$, where the $M$th iterate $\nabla_{\phi} \hat{\psi}_M(\phi)$ is used in \cref{eq:mceif}.\footnote{In \cref{sec:theo_monte}, we require that $\Delta_m = o\left(\sqrt{\frac{p\log p}{m}}\right)$.} Using a sequence of estimators in this way guarantees that the approximation error of  \cref{eq:mceif} is not dominated by {$\nabla_{\phi} \hat{\psi}_M(\phi)$. In practice, the target functional $\psi(\theta)$ might be quite complex, making gradient estimation challenging. For example, it might involve taking expectations with respect to conditionals of $p_{\phi}(x)$, or be defined implicitly as a solution to an optimization problem as in \citet{empirical-gateaux}.

One particularly simple and general way to address this challenge is to implement a Monte Carlo estimator of $\psi$ that can be transformed via automatic differentiation into an efficient Monte Carlo estimator for its gradient, a well-understood problem that is beyond the scope of this paper to review. We note that for the very wide class of functionals that can be written as nested expectations, recent work \cite{rainforth2018nesting,syed2023optimal,lew2023adev} gives formal statements of smoothness assumptions and theoretical results sufficient to obtain the oracle rate $\Delta_m$ in terms of numbers of samples, as well as algorithms that can be implemented using popular automatic differentiation software packages like PyTorch \cite{pytorch}.\footnote{For example code snippets of functionals, see \cref{A:code_ex}.}

\subsection{Theoretical Guarantees for MC-EIF} \label{sec:theo_monte}

We conclude by deriving a non-asymptotic error bound for how well \cref{eq:mceif} approximates \cref{eq:eif}. For fixed input dimension $D$ and model sizes $p$, \cref{eq:mceif} converges to \cref{eq:eif} at a $O_p(1/\sqrt{M})$ rate by the Law of Large Numbers. As we are interested in high-dimensional parametric families, we analyze the behavior of our approximation as a function of both input dimension $D$ and model size $p$. To prove our result, we use standard tools and assumptions from empirical process theory such as the requirement of sub-Gaussian tails~\cite{vaart-stats}.

\begin{assumption} \label{assum:subgaus}
Suppose  $x \sim p_{\phi}(x)$. Then,
there exists a universal constant $C_1$, $0 < \sigma_1 < \infty$, such that the normalized score vector $\tilde{x} \coloneqq \frac{1}{\sqrt{D}} \nabla_{\phi} \log p_{\phi}(x)$ is a sub-Gaussian random vector with parameter $C_1$.
\end{assumption}

As $\E[ \|\nabla_{\phi_j} \log p_{\phi}(x)\|^2_2] = O(D)$, $1 \leq j \leq p$, the division by $\sqrt{D}$ in \cref{assum:subgaus} ensures that the variance of the score does not grow unboundedly as $D \rightarrow \infty$. Thus, our assumption that $\tilde{x}$ is sub-Gaussian is mild. \cref{assum:bound_grad} below ensures that the functional and score are smooth enough by bounding their gradients. 

\begin{assumption} \label{assum:bound_grad}
There exist universal constants $C_2, C_3 < \infty$ such that $\|\nabla_{\phi} \psi(\phi) \|_F < C_2$ and $\left \|\frac{\nabla_{\phi} \log p_{\phi}(x^*)}{D} \right \|_2 < C_3$ for any $x^* \in \R^D$, for any $p$ and $D$.
\end{assumption}
Unlike our Monte Carlo approximation to the Fisher information matrix, we do not assume a particular type of  estimator for $\grad_{\phi} \psi(\phi)$. To prove convergence of MC-EIF to the true EIF, we assume that $\grad_{\phi} \hat{\psi}_M(\phi)$ converges to $\grad_{\phi} \psi(\phi)$ at the following rate:
\begin{assumption} \label{assum:conv_functional}
Let $\delta_M \coloneqq \grad_{\phi} \psi(\phi) -  \grad_{\phi} \hat{\psi}_M(\phi) \in \R^{L \times p}$ denote the approximation error. There exists a universal constant $C_{\psi} < \infty$ such that for $M > C_{\psi}$ and for any $\epsilon > 0$,
\begin{equation*}
    \mathbb{P}\left( \|\delta_M\|_F > \sqrt{\frac{p \log p + \epsilon}{M}} \right) < \exp(\text{-}\epsilon),
\end{equation*}
and
\begin{equation*}
    \mathbb{P}\left( \|\grad_{\phi} \hat{\psi}_M(\phi)\|_F > C_2 \right) < \exp(\text{-}\epsilon),
\end{equation*}
hold, where $C_2$ is defined in \cref{assum:bound_grad}.
\end{assumption}
In \cref{A:proof_mceif_convg}, we prove that Monte Carlo estimators of $\grad_{\phi} \psi(\phi)$ with gradient clipping \cite{grad-clip} satisfy \cref{assum:conv_functional}. Hence, \cref{assum:conv_functional} is a mild condition. Under these three assumptions, and the ones in \cref{thm:eif_fn}, we prove the following result in \cref{A:proof_mceif_convg}.

\begin{theorem} \label{thm:monte_eif_convg}
	Suppose \cref{assum:cont_diff_prob}, \cref{assum:cont_diff_func}, \cref{assum:fisher_inv}, \cref{assum:subgaus}, \cref{assum:bound_grad}, and \cref{assum:conv_functional} hold. Then, there exists universal constants $0 < C_4$ and $C_5 < \infty$, such that for any $\epsilon > 0$ and $M > \max(C_5(p + \epsilon)C_1^2, C_{\psi})$, 
\begin{equation}
    |\varphi_{\phi}(x^*) - \hat{\varphi}_{\phi, M}(x^*)| \leq C_4 \lambda_{\text{max}}(\inv{\Sigma})\sqrt{\frac{p \log p  + \epsilon}{M}},
\end{equation}
for $x^* \in \R^D$ with probability at least $1 - 2\exp(-\epsilon)$, where $\Sigma \coloneqq \text{cov}(\tilde{x})$ and $\lambda_{\text{max}}(\cdot)$ denotes the largest eigenvalue of a matrix.
\end{theorem}

\cref{thm:monte_eif_convg} tells us that $M$ needs to scale linearly with $p \log p$ to guarantee close pointwise approximation. %In \cref{thm:monte_eif_convg}, we assumed that $\grad \psi(\phi)$ is analytically known. %If $\grad \psi(\phi)$ is also approximated via Monte Carlo, then an application of the triangle inequality will result in a similar bound.

\subsection{An EIF Cookbook}

As a generalization of the gradient operator on ordinary functions, the EIF viewed as an operator on functionals can be shown to have a number of convenient algebraic properties \citep{kennedy-semi-emp-review, semiparametric_review}, many of which are inherited directly by the MC-EIF estimator. In \cref{A:eif_cookbook}, we highlight several ways in which these properties can be used to extend the basic MC-EIF framework (and the MC-EIF-based robust estimators in the next section) to new classes of models and functionals, significantly increasing the range of practical use cases addressable by an implementation of MC-EIF in a differentiable probabilistic programming language like Pyro \cite{pyro}.

\section{MC-EIF for Automated Efficient Inference}
\label{sec:programmable-DRI}

In \cref{thm:monte_eif_convg}, we proved that MC-EIF is close to the true efficient influence function pointwise. In this Section we; (i) show how MC-EIF can be used to automate the construction of popular efficient estimators, and (ii) prove how many Monte Carlo samples are needed to ensure that key statistical properties hold when MC-EIF is used instead of the true EIF. In doing so, MC-EIF brings conceptual clarity to the practice of constructing efficient estimators, and how these estimators can be implemented using existing differentiable probabilistic programming languages. 

All three of the efficient estimator templates we explore in this Section involve some combination of plug-in estimation and EIF-based computations. A key practical benefit of our work is that MC-EIF-based efficient estimators are entirely modular; advances in general-purpose probabilistic inference technology directly translate to advances in efficient estimation under our framework.

\subsection{Von Mises One Step Estimator}
We start with the simple \textit{Von Mises One Step Estimator}, which corrects the plug-in estimator in \cref{sec:problem-statement} by adding the average value of the efficient influence function on a held out dataset. Despite its simplicity, this estimator achieves optimal statistical rates~\cite{semiparametric_review}. Our one step estimator using MC-EIF ($\hat{\varphi}_{\phi, M}(x)$) instead of the true efficient influence function ($\varphi_{\phi}(x)$) is provided in \cref{alg:one_step}.

\textbf{Theoretical Guarantees.} We call the one step estimator that uses the true efficient influence function instead of MC-EIF in \cref{alg:one_step} the \emph{analytic one step estimator}. Below we prove how many Monte Carlo samples are needed to ensure our estimator for finite $M$ has the same statistical properties as the analytic one step estimator.
\begin{proposition} \label{prop:one_step_convg}
Suppose that the assumptions in \cref{thm:monte_eif_convg} hold. Further, suppose $M = O(Np \log p)$ and $p > O(\log{N})$. Then, \cref{alg:one_step} converges to $\theta^*$ at the same rate as the analytic one-step estimator.
\end{proposition}
By \cref{prop:one_step_convg}, if the analytic one step estimator achieves $O_p(\frac{1}{\sqrt{N}})$ convergence to $\theta^*$, then so does our estimator when $M = O(Np \log p)$ in \cref{alg:one_step}.

\subsection{Debiased/Double ML}
Next, we express debiased/double ML (DML)~\cite{dml} in terms of MC-EIF. To rewrite DML explicitly in terms of MC-EIF, we largely follow \citet{dml, influence_dml}.\footnote{DML constructs Neyman orthogonal corrected scores which are more general than efficient influence function corrected ones. Specifically, DML can handle functionals which are not pathwise differentiable. As we only consider parametric model families, however, we can assume without loss of generality that the functional is pathwise differentiable~\cite{orth_stat_learn}. Consequently, we can express DML in terms of efficient influence functions.} 

\textbf{Construction of Orthogonal Generalized Method of Moment (GMM) Estimators.} GMM-based estimators are defined by a nuisance functional $\eta(\cdot) \in \R^J$, $J \in \mathbb{N}$, and a set of $K \in \mathbb{N}$ functions $\{g_k(x, \eta(p_{\phi}), \theta)\}_{k=1}^K$, often called \emph{estimating equations}.\footnote{We only consider parametric models, therefore we can assume without loss of generality that the range space of the nuisance functional is finite-dimensional.} These estimating equations are selected so that their roots uniquely identify $\theta^*$ when the nuisance parameters $\eta(p_{\phi})$ are estimated correctly. That is,
\begin{equation} \label{eq:gmm}
	\E_{x \sim p_{\phi^*}(x)}[g(x, \eta(p_{\phi^*}), \theta)] = 0 \iff \theta = \theta^*,
\end{equation}

\begin{algorithm}[t]
\caption{MC-EIF one step estimator}\label{alg:one_step}
\begin{algorithmic}
\STATE {\bfseries Input:} Target functional $\psi$, initial estimate $\hat{\phi}$, held out datapoints $\{x_n\}_{n=N/2 + 1}^N$, \# Monte Carlo samples $M$
\STATE $\hat{\theta}_{\text{plug-in}} \gets \psi(\hat{\phi})$ \COMMENT{plug-in estimate}
\STATE $C = \frac{2}{N} \sum_{n=N/2+1}^N \hat{\varphi}_{\hat{\phi}, M}(x_n)$ 
\STATE {\bfseries Return:} $\hat{\theta}_{\text{plug-in}} + C$
\end{algorithmic}
\end{algorithm}

where $g \coloneqq (g_1, \cdots, g_K)$. As an example, $g$ might be the gradient of the log-likelihood function. To make GMM-based estimators less sensitive to incorrect estimation of the nuisance parameters, \citet{dml, influence_dml, auto_dml} replace $g(\cdot)$ with the \emph{orthogonal moment function}, constructed using influence functions. In our setting, the orthogonal moment function equals the following:
\begin{equation}
    g(x, \eta(p_{\phi}), \theta) +  \varphi_{\phi}(x, \theta),
\end{equation}
where $\varphi_{\phi}(x, \theta)$ is the efficient influence function associated with the functional $\mu_{\theta}(\phi) = \E_{x \sim p_{\phi}}[g(x, \eta(p_{\phi}), \theta)]$ for fixed $\theta$ by Equation 2.6 in \citet{influence_dml}. By \cref{thm:eif_fn},
\begin{equation} \label{eq:inf_gmm}
      \varphi_{\phi}(x, \theta) \coloneqq [\nabla_{\phi} \mu_{\theta}(\phi)]^T \inv{I(\phi)} \nabla_{\phi} \log p_{\phi}(x).
\end{equation}
\begin{algorithm}[t!]
\caption{MC-EIF debiased ML w/ simple data splitting}\label{alg:dml}
\begin{algorithmic}
\STATE {\bfseries Input:} Vector of estimating equations $g$, initial estimate $\hat{\phi}$, held out datapoints $\{x_n\}_{n=N/2+1}^N$, \# of Monte Carlo samples $M$
\STATE $f(\theta) \gets \frac{2}{N} \sum_{n=N/2 + 1}^N \left[ g(x_n, \eta(p_{\hat{\phi}}), \theta) +  \hat{\varphi}_{\hat{\phi}, M}(x_n, \theta)  \right]$ %\COMMENT{empirical orthogonal moment function}
\STATE {\bfseries Return:} $\{\theta: f(\theta) = 0 \}$
\end{algorithmic}
\end{algorithm}

\begin{algorithm}[t!]
\caption{MC-EIF one step TMLE}\label{alg:meta_tmle_one_step}
\begin{algorithmic}
\STATE {\bfseries Input:} Target functional $\Psi$, initial estimate $\hat{\phi}$, held out datapoints $\{x_n\}_{n=N/2 + 1}^N$, \# Monte Carlo samples $M$
\STATE $p(\epsilon, x) = (1 + \epsilon^T \hat{\varphi}_{\hat{\phi}, M}(x)) p_{\hat{\phi}}(x)$
\STATE $\hat{\epsilon} = \arg \max_{\epsilon \in \R^L: p(\epsilon, x) \in \mathcal{M}} \frac{2}{N} \sum_{n=1}^N \log p(\epsilon, x_n)$
\STATE {\bfseries Return:} $\Psi(p(\hat{\epsilon}, \cdot))$
\end{algorithmic}
\end{algorithm}

Since $g$ is a known by assumption, we can readily use the Monte Carlo methods in~\citet{auto-encode-vb, stoch-comp-graphs} to automatically approximate $\nabla_{\phi} \mu_{\theta}(\phi)$. We summarize the DML algorithm in \cref{alg:dml} which replaces \cref{eq:inf_gmm} with our MC-EIF approximation.\footnote{In DML, the authors use $K$-fold cross-fitting instead of the simple data split described here for better finite sample properties.}

\textbf{Theoretical Guarantees.} For general estimating equations, it is difficult to quantity how errors in our MC-EIF approximation to \cref{eq:inf_gmm} lead to changes in final estimates. When the estimating equations have more structure, however, we obtain a similar result as in \cref{prop:one_step_convg}.
\begin{assumption} \label{assum:linear_gmm} $g(x_n, \eta(p_{\phi}), \theta) = m(x_n, \eta(p_{\phi})) - \theta$ for some vector of known functions $m(\cdot)$.
\end{assumption}
\cref{assum:linear_gmm} was made in several works~\cite{auto_dml, influence_dml} already. We prove an analogous rate guarantee as in \cref{prop:one_step_convg} under \cref{assum:linear_gmm}.
\begin{proposition} \label{prop:doubl_ml_convg}
	Suppose that the assumptions in \cref{thm:monte_eif_convg} and \cref{assum:linear_gmm} hold. Further, suppose $M = O(Np \log p)$ and $p > O(\log{N})$. Then, \cref{alg:dml} converges to $\theta^*$ at the same rate as the analytic DML estimator, namely the one that sets $M = \infty$ in \cref{alg:dml}.
\end{proposition}

\subsection{Targeted Minimum Loss Estimation}

We conclude by writing targeted minimum loss estimation (TMLE)~\cite{tmle} explicitly in terms of MC-EIF. Unlike the one step estimator or DML, TMLE directly corrects the estimated distribution $p_{\hat{\phi}}(x)$ and then plugs in the corrected distribution into the functional $\Psi$ as the final estimate. To perform this correction it perturbs $p_{\hat{\phi}}$ in the direction of the influence function, searching for the optimal step size by maximizing the perturbed likelihood on the held out dataset. Intuitively, TMLE can be viewed as a form of gradient ascent in function space. We show one step TMLE~\cite{tmle} in \cref{alg:meta_tmle_one_step}. The multi-step TMLE version is computed by iterating \cref{alg:meta_tmle_one_step} multiple times until $\epsilon$ approximately equals 0~\cite{tmle}. % In \cref{A:tmle}, we also show how to implement MC-EIF TMLE in practice. 

\section{Experiments}
\label{sec:experiments}

We start by comparing the quality of MC-EIF against other methods for influence function approximation. Then, we show how
MC-EIF behaves when; (i) the number of Monte Carlo samples is varied, (ii) the dimensionality of the input is varied, and (iii) the efficient estimator type is varied. Our empirical results ultimately validate our theoretical results in \cref{sec:MC-EIF} and \cref{sec:programmable-DRI}. Finally, we show how MC-EIF can be used to automate the construction of efficient estimators for new functionals by revisiting a classic problem in optimal portfolio theory. Our MC-EIF implementation is publicly available in the Python package ChiRho. All results shown here are end-to-end reproducible.

\textbf{Comparison to \citet{empirical-gateaux}.}
In \citet{empirical-gateaux}, the authors target the nonparametric influence function, which is the unique influence function when $\mathcal{M} = L^2(\mathbb{P})$ in \cref{def:gateaux}. By contrast, we target the efficient influence function. Thus, for evaluation, we compare how well the empirical Gateaux method from \citet{empirical-gateaux} approximates the nonparametric influence function and how well our MC-EIF method approximates the efficient influence function on the same data-generating process. 

To have a ground truth for comparison, we select a simple model and functional where we can analytically compute the nonparametric and efficient influence functions. To this end, we consider the problem of estimating the expected density, $\Psi(\mathbb{P}) = \int \mathbb{P}(x)^2 dx$ as in \citet{bickel-density, computer-eif}. We further suppose that $x \sim N(\mu, \sigma)$. We consider two parametric model families: one where $\mu$ is unknown but $\sigma=1$, and one where both $\mu$ and $\sigma$ are unknown, which we call $\mathcal{M}_1$ and $\mathcal{M}_2$ respectively. As the nonparametric influence function makes no assumptions on the underlying model family, it remains fixed across $\mathcal{M}_1$ and $\mathcal{M}_2$ and always equals $2(\mathbb{P}(X) - \Psi(\mathbb{P}))$~\cite{computer-eif}. However, the EIF equals zero for $\mathcal{M}_1$, as the expected density does not depend on $\mu$. Hence, any plug-in estimate for models in $\mathcal{M}_1$ will result in a correct value of the expected density, and thus no distributional perturbations produce any change. In $\mathcal{M}_2$, the efficient influence function for the expected density depends on the unknown $\sigma$. See \cref{fig:analytic_inf_density} in the Appendix for a visualization of the analytic influence functions.

\begin{figure}[tbp]
  \centering
  \begin{subfigure}[t]{0.38\columnwidth}
    \vspace{-5cm}
    \includegraphics[width=\linewidth]{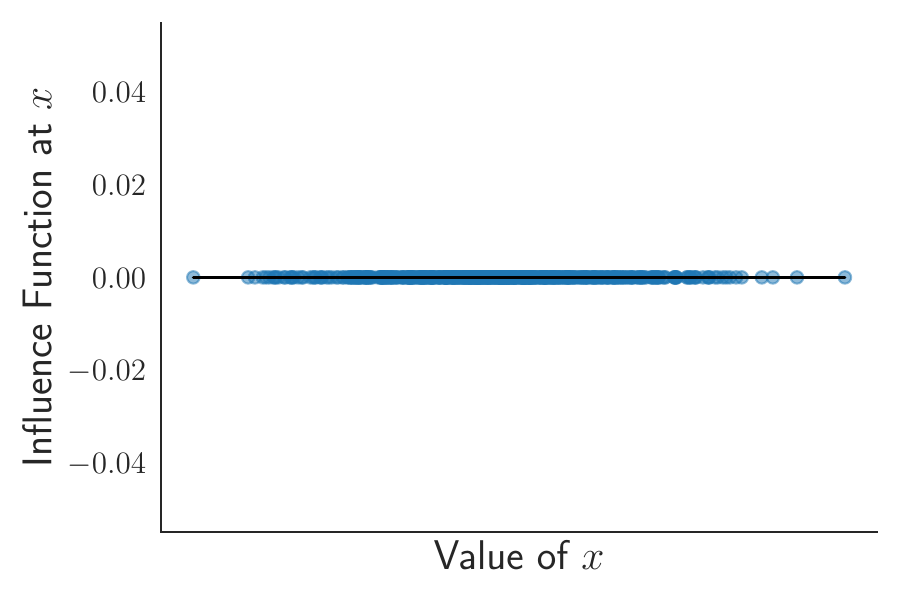}
    \vspace{-6mm}
    \caption{MC-EIF, known $\sigma^2$}
    \vspace{2mm}
    \includegraphics[width=\linewidth]{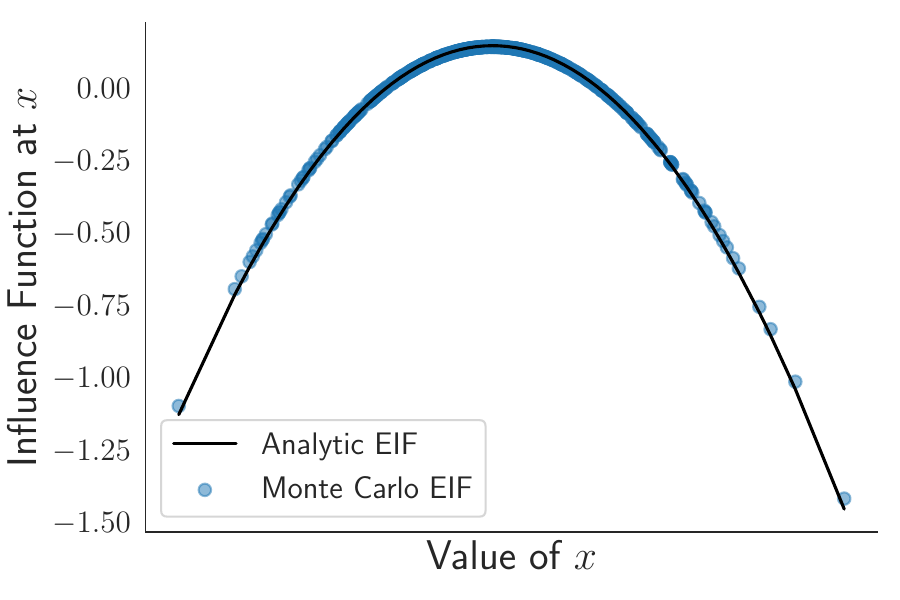}
    \vspace{-4mm}
    \caption{MC-EIF, unknown $\sigma^2$}
  \end{subfigure}
  \hfill
  \begin{subfigure}[t]{0.6\columnwidth}
    \includegraphics[width=\linewidth]{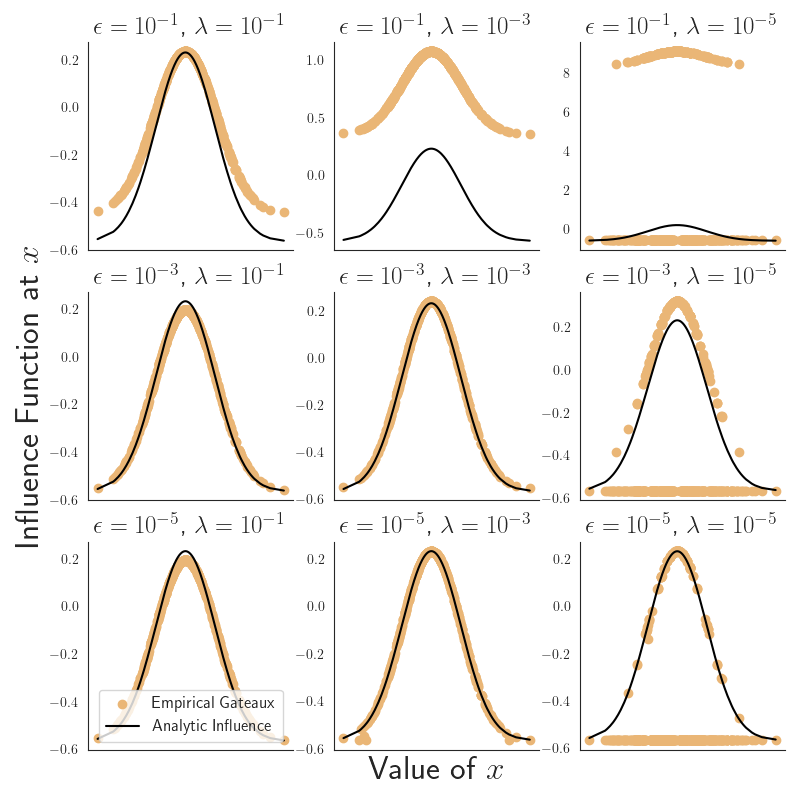}
    \label{fig:figure1}
    \vspace{-6mm}
    \caption{Empirical Gateaux}
  \end{subfigure}
  
  \caption{\textbf{Comparison between MC-EIF and empirical Gateaux approximation.} MC-EIF (a and b) is less sensitive to hyperparameters parameters ($\epsilon$ and $\lambda$) than the empirical Gateaux baseline (c).} \label{fig:emp_gat_v_mc_eif}
\end{figure}

\cref{fig:emp_gat_v_mc_eif} summarizes how well the empirical Gateaux derivative method approximates the nonparametric influence function and how well our MC-EIF method approximates the EIF at the point $p_{\phi} = N(0, 1)$. We see that MC-EIF is able to approximate the efficient influence function very well ($M = 10^4$ samples). By contrast, the empirical Gateaux derivative is highly sensitive to the choice of two kernel smoothing hyperparameters, $\epsilon$ and $\lambda$. As the true influence function is not known, it is unclear how to select $\epsilon$ and $\lambda$. Such numerical instability was already discussed in~\citet{computer-eif}, where the precision necessary must get exponentially smaller with input dimension, making it infeasible when $D \approx 10$. For us, we only have a single tunable parameter ($M$), where larger $M$ unambiguously provides a better approximation. We illustrate how the error decays as a function of $M$ in \cref{fig:monte_carlo_increase_density}. We provide an extended discussion of the automation limitations of the empirical Gateaux method in \cref{A:automatic_comp}. We attempted to use the empirical Gateaux derivative method as a baseline for other experiments, but were unable to achieve numerically stable solutions for any $p>2$ without prohibitively long run-times.

\begin{figure}[t]
  \centering\includegraphics[width=0.8\linewidth]{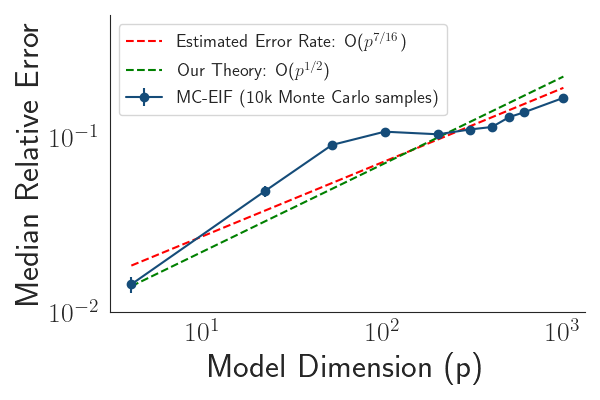}
    \vspace{-3mm}
    \caption{\textbf{Empirical evidence for convergence theory.} Increasing $p$ for the average treatment effect experiments produces MC-EIF approximation errors that closely match ~\cref{thm:monte_eif_convg}.}
    \vspace{-0.5cm}
\end{figure}

In the next two Subsections, we focus on a classic model consisting of a binary treatment, high-dimensional continouous confounders, and Gaussian distributed response; see \cref{A:experiment_details} for the precise model formula. We assume that the analyst is interested in estimating the average treatment effect (ATE), where the true ATE is zero but unknown. All influence function computations are relative to an initial point estimate $\hat{\phi}$, found through maximum a posteriori estimation using 500 training datapoints. Due to the exponential runtime in dimension for the methods in \citet{empirical-gateaux, computer-eif}, we focus on comparing MC-EIF with the analytic influence function for ATE below.

\textbf{Sensitivity to Dimensionality.} \cref{thm:monte_eif_convg} implies that for a fixed number of Monte Carlo samples $M$, the quality of the approximation degrades with the square root of model dimension $p$. In \cref{fig:dim_sens}, we empirically show how the quality of the approximation degrades as $p$ increases for $M = 10^4$ fixed. Based on \cref{fig:dim_sens}, the empirical results closely match the theoretical behavior predicted by \cref{thm:monte_eif_convg}.\footnote{As discussed in \cref{sec:theo_monte}, to make the error not grow with $p$, we would need $M \asymp p \log p$.} We also show how the computational complexity of MC-EIF scales as $p$ increases in \cref{fig:dim_sens}.

\textbf{Sensitivity to Estimator Type.} Here, we consider a high-dimensional setup where there are 200 confounders but only 500 training datapoints. We simulate 100 different datasets with this configuration to approximate the sampling distribution of different efficient estimators. In \cref{fig:estimator_down_error}, we see that across estimators, using MC-EIF instead of the true EIF results in minimal downstream error. This is consistent with our theoretical results in \cref{sec:programmable-DRI}. 

\begin{figure}[tbp]
  \centering
  \begin{subfigure}{0.32\columnwidth}
    \includegraphics[width=\linewidth]{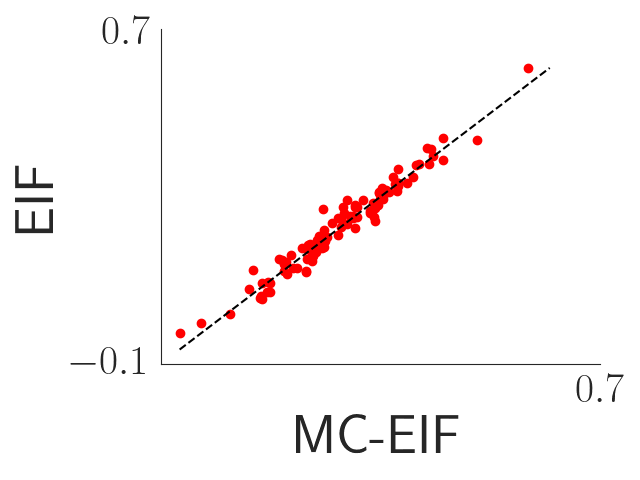}
    \caption{One Step}
  \end{subfigure}
  \hfill
  \begin{subfigure}{0.32\columnwidth}
    \includegraphics[width=\linewidth]{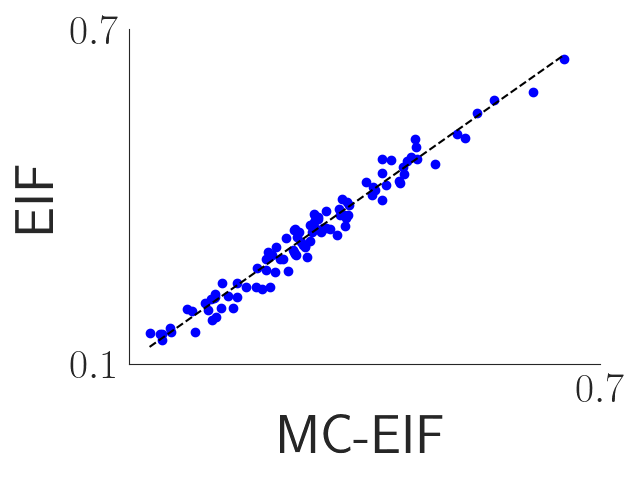}
    \caption{TMLE}
  \end{subfigure}
  \hfill
  \begin{subfigure}{0.32\columnwidth}
    \includegraphics[width=\linewidth]{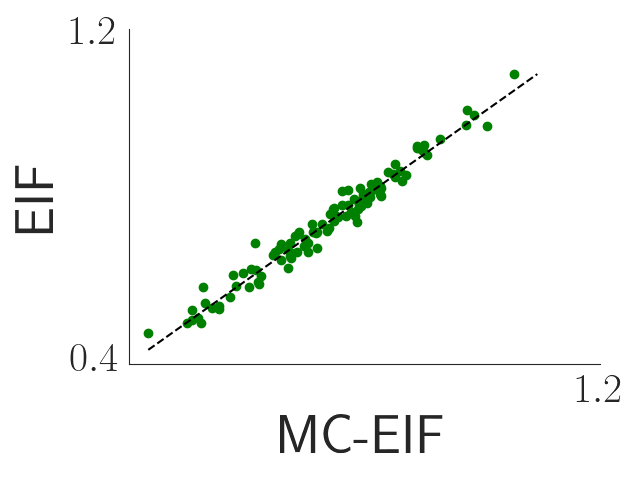}
    \caption{Double ML}
  \end{subfigure}
  \vspace{-4mm}
  \caption{\textbf{Comparison between ATE estimators using MC-EIF and analytic EIF.} MC-EIF produces ATE estimates that are very close to the diagonal, representing an oracle estimator of the EIF.} \label{fig:estimator_down_error}
\end{figure}

In \cref{fig:estimator_comp}, we find that the one step estimator empirically performs the best (over a 30\% improvement relative to the plug-in estimator). Debiased/double ML performs the worst but that might be due to not using $K$ fold cross-fitting. Nevertheless, the point of the current work is agnostic to the choice of efficient estimator. We simply show that MC-EIF can be used to automate the construction of efficient estimators for various functionals.

\textbf{Ability to Handle New Functionals.} To illustrate MC-EIF's flexibility, we revisit a classic problem in optimal portfolio theory. Suppose that $x \in \R^D$ is a vector of asset returns. We are interested in estimating the optimal portfolio weights $\theta^* \in \R^D$ that maximize the expected return while minimizing the variance of the portfolio. Then, the Markowitz optimal portfolio optimization problem~\cite{markowitz-theory} is found by solving the following convex optimization problem:
\begin{equation} \label{eq:markowitz}
	\begin{split}
	& \Psi_{\lambda}(p_{\phi}) = \arg \max_{\theta \in \R^D} \theta^T \E_{p_{\phi}}[x] - \lambda \theta^T \text{Cov}(x; p_{\phi}) \theta, \\
		& \qquad \text{subject to } \sum_{i=1}^D \theta_i = 1,
	\end{split}
\end{equation}
where $\lambda$ is the tradeoff between expected return and variance (measure of risk), and $\text{Cov}(x; p_{\phi})$ denotes the covariance matrix with respect to $p_{\phi}$. Hence, the optimal weights functional $\Psi_{\lambda}(p_{\phi})$ depends on a high-dimensional nuisance, namely the $D \times D$ covariance matrix of returns. The target $\theta^*_{\phi, \lambda} = \Psi_{\lambda}(p_{\phi})$ is a much lower $D$-dimensional target parameter. A canonical choice in the financial econometrics literature is setting $\lambda = \infty$, which corresponds to the \emph{global minimum variance portfolio}~\cite{cap-weighted-markowitz, ma-markowitz, cov-est-mtp2}.

We assume $x$ is drawn from a multivariate Gaussian distribution with unknown covariance matrix for $D=25$ and $N=1000$ datapoints. We randomly sample the true covariance matrix using a Lewandowski-Kurowicka-Joe distribution on positive definite matrices. To our knowledge, the is no literature on an efficient estimator for the minimum variance portfolio functional. Armed with MC-EIF, computing an efficient estimate for the global minimum variance portfolio functional is automatic. We evaluate MC-EIF and the one step estimator using the relative expected volatility (REV) and the root mean-squared-error (RMSE) between the estimated  and the true optimal portfolio weights. Here, the expected volatility is calculated by applying the estimated weights with the actual covariance to the objective in \cref{eq:markowitz}. Repeating our experiment using 50 randomly generated datasets, we find that MC-EIF enables substantially improved estimates, as shown in \cref{tab:Markowitz}.

\begin{table}[]
    \centering
    \begin{tabular}{c|c|c}
        & One Step MC-EIF & Plug-in\\
        \hline \hline
        REV & $\mathbf{1.86 \pm .35}$ & $2.60 \pm .35$\\
        RMSE & $\mathbf{.08 \pm .02}$ & $.14 \pm .02$\\
        \hline
    \end{tabular}
    \caption{\textbf{Empirical results for Markowitz optimal portfolio optimization.} Using MC-EIF, \cref{alg:one_step} estimates weights that achieve lower relative expected volatity (REV) and root mean-squared-error (RMSE) compared to the oracle estimator.}
    \label{tab:Markowitz}
\end{table}

%This example highlights the ease of using MC-EIF to automate the construction of efficient estimators for new functionals.

% \begin{figure}[t]
% 	\centering
% 	\includegraphics[width=4.5cm]{./exp_figs/markowitz_optimal.png}
% 	\caption{Median absolute error between the optimal weights (i.e., when true covariance matrix is used in \cref{eq:markowitz}) versus weights found from plug-in and one step estimators. Distribution over 50 simulated datasets.} \label{fig:markowitz}
% \end{figure}

\section{Conclusion}
\label{sec:conclusion}
We have shown both theoretically and empirically that MC-EIF can reliably be used to automate efficient estimation. Our key contributions include MC-EIF's consistency and capability to achieve optimal convergence rates. Empirical evidence shows that MC-EIF performs comparably to traditional estimators using analytic EIFs. Additionally, we illustrate the practical application of MC-EIF in scenarios where the analytic EIF is not known.

%One is relating the Von Mises one step estimator to natural gradient descent. In particular, when $\psi(\phi)$ is the identity map, these two methods perform the same computation by \cref{thm:eif_fn}. \TD{Eli - can you add some of your thoughts here}

% It would be interesting to explore the theoretical implications for robust estimation of solving subproblems in the MC-EIF estimator approximately using more scalable procedures drawn from the large body of related work on similar problems, e.g. low-rank approximations to the Fisher information used for natural gradient descent, biased estimates of the log-likelihood and its derivatives in variational inference, or stochastic optimization algorithms for the inverse Hessian-vector product \TD{CITATIONS}.

There is also a growing literature drawing connections between semiparametric theory and many seemingly disparate heuristic methods in deep learning~\cite{vowels2022free, bae2022if, dao2021knowledge, zhu2023doubly}. MC-EIF may help strengthen these connections from both sides: perhaps there are more existing heuristic methods that could be understood in a unified way as semiparametric computations via particular instantiations of MC-EIF; meanwhile MC-EIF's broad applicability may aid in the translation of other principled semiparametric methods into practical tools for real-world machine learning problems.

%For example, the update rule in natural gradient descent has the same form as \cref{eq:eif} when the functional is the identity map~\cite{nat-grad-descent}.

\newpage

\section{Acknowledgements}

The authors would like to thank DARPA for funding this work through the Automating Scientific Knowledge Extraction and Modeling (ASKEM) program, Agreement No. HR0011262087. The views, opinions and/or findings expressed are those of the authors and should not be interpreted as representing the official views or policies of the Department of Defense or the U.S. Government. The authors would also like to thank Tamara Broderick, David Burt, and Ryan Giordano for helpful discussions.

% \section*{Broader Impacts}

% We believe that the following statement taken verbatim from the \href{https://icml.cc/Conferences/2024/CallForPapers}{ICML call for papers} represents the broader societal impacts of our work:

% \begin{quote}
%     This paper presents work whose goal is to advance the field of Machine Learning. There are potential societal consequences of our work, none of which we feel must be significantly highlighted here.
% \end{quote}

\bibliography{references}
\bibliographystyle{icml2024}

%%%%%%%%%%%%%%%%%%%%%%%%%%%%%%%%%%%%%%%%%%%%%%%%%%%%%%%%%%%%%%%%%%%%%%%%%%%%%%%
%%%%%%%%%%%%%%%%%%%%%%%%%%%%%%%%%%%%%%%%%%%%%%%%%%%%%%%%%%%%%%%%%%%%%%%%%%%%%%%
% APPENDIX
%%%%%%%%%%%%%%%%%%%%%%%%%%%%%%%%%%%%%%%%%%%%%%%%%%%%%%%%%%%%%%%%%%%%%%%%%%%%%%%
%%%%%%%%%%%%%%%%%%%%%%%%%%%%%%%%%%%%%%%%%%%%%%%%%%%%%%%%%%%%%%%%%%%%%%%%%%%%%%%
\newpage
\appendix
\onecolumn
\section{Proofs} \label{A:proofs}

\subsection{Proof of \cref{thm:monte_eif_convg}} \label{A:proof_mceif_convg}
\begin{proof}
	Let $\{x_m\}_{m=1}^M$ denote the $M$ Monte Carlo samples in \cref{eq:emp_fisher}, and let $\tilde{x}_m \coloneqq \frac{1}{\sqrt{D}} \nabla_{\phi} \log p_{\phi}(x_m)$ for $1 \leq m \leq M$. Let $\hat{\Sigma} = \frac{1}{M} \sum_{m=1}^M \tilde{x}_m \tilde{x}_m^T$ denote the sample covariance matrix. Then, $\Sigma = \frac{1}{D} I(\phi)$ and $\hat{\Sigma} = \frac{1}{D} \hat{I}(\phi)$. Hence,
	\begin{equation}
        \begin{split}
		 |\varphi_{\phi}(x) - \hat{\varphi}_{\phi, M}(x)| &=  \left | [\nabla_{\phi}\hat{\psi}_M(\phi)]^T (\inv{\Sigma} - \inv{\hat{\Sigma}})  \frac{\nabla_{\phi} \log p_{\phi}(x^*)}{D} + \delta_M \inv{\Sigma} \frac{\nabla_{\phi} \log p_{\phi}(x^*)}{D} \right| \\
            &\leq  \left | [\nabla_{\phi}\hat{\psi}_M(\phi)]^T (\inv{\Sigma} - \inv{\hat{\Sigma}})  \frac{\nabla_{\phi} \log p_{\phi}(x^*)}{D} \right| + \left | \delta_M \inv{\Sigma} \frac{\nabla_{\phi} \log p_{\phi}(x^*)}{D} \right| \\
            &\leq \left | [\nabla_{\phi}\hat{\psi}_M(\phi)]^T (\inv{\Sigma} - \inv{\hat{\Sigma}})  \frac{\nabla_{\phi} \log p_{\phi}(x^*)}{D} \right| + \left | \delta_M \inv{\Sigma} C_3 \right| \\
            & \leq \left | [\nabla_{\phi}\hat{\psi}_M(\phi)]^T (\inv{\Sigma} - \inv{\hat{\Sigma}})  \frac{\nabla_{\phi} \log p_{\phi}(x^*)}{D} \right| + C_3 \lambda_{\text{max}}(\inv{\Sigma}) \|\delta_M\|_F
        \end{split}
	\end{equation}
By \cref{assum:conv_functional}, $\|\delta_M\|_F < \sqrt{\frac{p + \epsilon}{M}}$ with probability greater than $1 - \exp(-\epsilon)$ when $M > C_{\psi}$. If we can prove that there exists a constant $C_4$
\begin{equation}
    \left | [\nabla_{\phi}\hat{\psi}_M(\phi)]^T (\inv{\Sigma} - \inv{\hat{\Sigma}})  \frac{\nabla_{\phi} \log p_{\phi}(x^*)}{D} \right| \leq C_4 \lambda_{\text{max}}(\inv{\Sigma})\sqrt{\frac{p\log{p} + \epsilon}{M}}
\end{equation}
with probability greater than $1 - 2\exp(-\epsilon)$, the claim follows by an application of the union bound. By Theorem 10 in \citet{cov_bounds}, the claim follows if we can prove that there exists a universal constant $C_4$ such that 
	\begin{equation}
		\left \| \frac{\nabla_{\phi} \log p_{\phi}(x^*)}{D}^T  \inv{\Sigma} \tilde{x} \right \|_{\psi_2} \left \| \nabla_{\phi}\hat{\psi}_M(\phi)^T  \inv{\Sigma} \tilde{x} \right \|_{\psi_2} < C_4 \lambda_{\text{max}}(\inv{\Sigma}),
	\end{equation}
	with probability greater than $1 - \exp(-\epsilon)$, where $\|\cdot \|_{\psi_2}$ denotes the Orlicz sub-Gaussian norm; see Equation 9 in \citet{cov_bounds} for a precise definition of the Orlicz norm of a random vector.	With probability greater than $1 - \exp(-\epsilon)$, $\|\nabla_{\phi}\hat{\psi}_M(\phi)\|_F < C_2$ by \cref{assum:conv_functional}. Hence, with probability greater than $1 - \exp(-\epsilon)$,
		\begin{equation}
			\begin{split}
				\left \| \frac{\nabla_{\phi} \log p_{\phi}(x^*)}{D}^T  \inv{\Sigma} \tilde{x} \right \|_{\psi_2} \left \| \nabla_{\phi}\hat{\psi}_M^T  \inv{\Sigma} \tilde{x} \right \|_{\psi_2} &\leq C_2 C_3 \left \|  \inv{\Sigma} \tilde{x} \right \|_{\psi_2}^2 \\
					&\leq C_2 C_3 \| \Sigma^{-1/2} \|_2^2 \| \Sigma^{-1/2} \tilde{x} \|_{\psi_2}^2 \\
					&\leq  C_1 C_2 C_3 \|\Sigma^{-1/2}\|_2^2 \| \text{cov}(\Sigma^{-1/2} \tilde{x}) \|_2 \\
					&= C_1 C_2 C_3 \| \Sigma^{-1/2} \|_2^2 \\
					&= C_1 C_2 C_3 \lambda_{\text{max}}(\inv{\Sigma}),
			\end{split}
		\end{equation}
	where the first inequality follows from \cref{assum:bound_grad}, the third by \citet{subspace_bound} and \cref{assum:subgaus}, and last by the definition of the spectral norm of a matrix. The result now follows by setting $C_4 = C_1 C_2 C_3$.
\end{proof}

\textbf{\cref{assum:conv_functional} Holds for Monte Carlo Estimators.} Here we show if $\hat{\psi}_M$ is also approximated with $M$ Monte Carlo samples, then \cref{assum:conv_functional} holds. To this end, suppose
\begin{equation}
    \grad_{\phi} \hat{\psi}_M(\phi) = \frac{1}{M}\sum_{m=1}^M \grad_{\phi} g_{\phi}(w_m), \quad w_m \overset{\text{iid}}{\sim} q(w) \quad 1 \leq m \leq M, \quad \text{s.t.} \quad \E[\grad_{\phi} g_{\phi}(w_m)] = \grad_{\phi} \psi(\phi),
\end{equation}
for some distribution $q(w)$ and function $g_{\phi}$. Such a decomposition exists, for example, when the functional is expressible as a stochastic computation graph \cite{stoch-comp-graphs} or for reparameterizable densities \cite{auto-encode-vb}. Suppose further that there exists a universal constant such that $\grad_{\phi_j} g_{\phi}(w_m) \in \R^L$ is a sub-Gaussian random vector with parameter $\sigma_\psi$ for $1 \leq j \leq p$. Then, 
\begin{equation}
    \begin{split}
            \|\delta_M \|_F &= \sqrt{\sum_{j=1}^p\sum_{l=1}^L ([(\grad_{\phi_j} g_{\phi}(w_m)]_l - [(\grad_{\phi_j} \psi(w_m)]_l )^2} \\
            &\leq \sqrt{pL} \max_{1 \leq l \leq L, 1 \leq j \leq p} |[(\grad_{\phi_j} g_{\phi}(w_m)]_l - [(\grad_{\phi_j} \psi(w_m)]_l |
    \end{split}
\end{equation}
By Exercise 2.12 in \citet{wainwright-high-stat}, 
\begin{equation}
    \max_{1 \leq l \leq L, 1 \leq j \leq p} |[(\grad_{\phi_j} g_{\phi}(w_m)]_l - [(\grad_{\phi_j} \psi(w_m)]_l | = O_p\left(\sigma_\psi \sqrt{\frac{\log(pL)}{M}}\right)
\end{equation}
Hence, since $L$ is a constant, $\|\delta_M \|_F = O_p\left(\sqrt{\frac{p\log(p)}{M}}\right)$. Thus, the first equation in \cref{assum:conv_functional} holds. Under \cref{assum:bound_grad}, the second equation in \cref{assum:conv_functional} trivially holds using gradient clipping with $C_2$.

\subsection{Proof of \cref{prop:one_step_convg}}
\begin{proof}

By the Law of Large Numbers, $\frac{2}{N} \sum_{n=N/2+1}^N \varphi_{\hat{\phi}}(x_n)$ converges to $\E_{x \sim p_{\phi^*}(x)}[\varphi_{\hat{\phi}}(x)]$ at a $O_p\left(\frac{1}{\sqrt{N}}\right)$ rate. Hence, it suffices to prove that difference between the analytic one step estimator  and \cref{alg:one_step} with finite $M$ decays at a $O_p\left(\frac{1}{\sqrt{N}}\right)$ rate. To this end, their difference equals
\begin{equation}
	\left | \frac{2}{N} \sum_{n=N/2+1}^N \left(\varphi_{\hat{\phi}}(x_n) - \hat{\varphi}_{\hat{\phi}, M}(x_n) \right) \right | \leq \max_{n=N/2 + 1, \cdots, N}  \left | \varphi_{\hat{\phi}}(x_n) - \hat{\varphi}_{\hat{\phi}, M}(x_n) \right |.
\end{equation}
Let $\epsilon > 0$ and suppose $M = C_4^2 \max(C_5, 1)(p \log(p) + \epsilon)\max(C_1^2, 1) N \lambda_{\text{max}}^2(\inv{\Sigma}) = O(pN)$, where $C_1$ is defined in \cref{assum:subgaus}, and  $C_5$ and $\inv{\Sigma}$ are defined in \cref{thm:monte_eif_convg}. Then, it suffices to prove for this choice of $M$ that
\begin{equation}
	 \ptrue \left(\max_{n=N/2 + 1, \cdots, N}  \left | \varphi_{\hat{\phi}}(x_n) - \hat{\varphi}_{\hat{\phi}, M}(x_n) \right | >  \sqrt{\frac{2}{N}} \right) \leq \exp(-\epsilon).
\end{equation}
Let $\epsilon_N = \epsilon + \log(N/2)$. By \cref{thm:monte_eif_convg} and the union bound,
\begin{equation} \label{eq:bound_helper1}
% \resizebox{1.1\textwidth}{!}
%      {
	\begin{split}
	 & \ptrue \left(\max_{n=N/2 + 1, \cdots, N}  \left | \varphi_{\hat{\phi}}(x_n) - \hat{\varphi}_{\hat{\phi}, M}(x_n) \right | >  C_4 \lambda_{\text{max}}(\inv{\Sigma})\sqrt{\frac{p \log(p) + \epsilon_N}{M}} \right) \\
  &\leq \frac{N}{2} \sum_{n=N/2 + 1}^N \ptrue \left( \left | \varphi_{\hat{\phi}}(x_n) - \hat{\varphi}_{\hat{\phi}, M}(x_n) \right | >  C_4 \lambda_{\text{max}}(\inv{\Sigma})\sqrt{\frac{p \log(p) + \epsilon_N}{M}} \right) \\
	 	&\leq  N/2 \exp(-\epsilon_N) \\
	 	& = \exp(-\epsilon_N + \log N/2) \\
	 	& =  \exp(-\epsilon)
	 \end{split}
 % }
\end{equation}
Now,
\begin{equation} \label{eq:bound_helper2}
	\begin{split}
	C_4 \lambda_{\text{max}}(\inv{\Sigma})\sqrt{\frac{p \log(p) + \epsilon_N}{M}} &= C_4 \lambda_{\text{max}}(\inv{\Sigma})\sqrt{\frac{p \log(p) + \epsilon + \log(N/2)}{C_4^2 \max(C_5, 1)(p + \epsilon)\max(C_1^2, 1) N  \lambda_{\text{max}}^2(\inv{\Sigma})}} \\
		& \leq \sqrt{\frac{p \log(p) + \epsilon + \log(N/2)}{(p \log(p) + \epsilon) N }} \\
		& \leq \sqrt{\frac{2}{N}}.
	\end{split}
\end{equation}
The proof now follows from \cref{eq:bound_helper1} and \cref{eq:bound_helper2}.

\end{proof}

\subsection{Proof of \cref{prop:doubl_ml_convg}}

\begin{lemma} \label{lemma:dml_vs_one_step}
Suppose \cref{assum:linear_gmm} holds. Then, $\hat{\theta}_{\text{DML}} = \frac{2}{N} \sum_{n=N/2+1}^N  m(x_n, \eta(p_{\hat{\phi}})) + \sum_{n=1}^N  \varphi_{\hat{\phi}}(x_n)$, where $\varphi_{\phi}(x)$ is the influence function associated with the functional $\psi(\phi) = \E_{x \sim p_{\phi}(x)}[m(x_n, \eta(p_{\phi}))]$.
\end{lemma}

\begin{proof}
	We claim $\varphi_{\phi}(x, \theta) = \varphi_{\phi}(x)$ for all $\theta$. To prove this claim, notice that $\nabla_{\phi} \mu_{\theta}(\phi) = \nabla_{\phi} \mu_{\theta^{\prime}}(\phi)$ for arbitrary $\theta$ and $\theta^{\prime}$ since $\mu_{\theta}(\phi) = \E_{x \sim p_{\phi}}[m(x_n, \eta(p_{\phi}))] - \theta$. Hence, the claim follows from \cref{eq:inf_gmm}. Now,
   \begin{equation}
	   \frac{2}{N} \sum_{n=N/2 + 1}^N \left[ g(x_n, \eta(p_{\hat{\phi}}), \theta) +  \varphi_{\hat{\phi}}(x_n, \theta) \right] = \frac{1}{N} \sum_{n=1}^N  \left[ m(x_n, \eta(p_{\hat{\phi}})) +  \varphi_{\hat{\phi}}(x_n) \right]  - \theta.
   \end{equation}
   Hence, $\hat{\theta}_{\text{DML}} = \frac{2}{N} \sum_{n=N/2 + 1}^N  m(x_n, \eta(p_{\hat{\phi}})) + \sum_{n=1}^N  \varphi_{\hat{\phi}}(x_n)$.
\end{proof}

\begin{proof}[Proof of \cref{prop:doubl_ml_convg}]
	By \cref{lemma:dml_vs_one_step}, \cref{alg:dml} uses the same correction term $C$ in \cref{alg:one_step}. Hence, the proof of \cref{prop:doubl_ml_convg} now follows from \cref{prop:one_step_convg}.
\end{proof}

\begin{remark}
	By \cref{lemma:dml_vs_one_step}, the only difference between \cref{alg:dml} and  \cref{alg:one_step} is a different value for the initial estimate of $\theta^*$. Specifically, in DML, the initial estimate of $\theta^*$ is $\frac{2}{N} \sum_{n=N/2+1}^N  m(x_n, \eta(p_{\hat{\phi}}))$, which averages over datapoints drawn from the true distribution. By contrast, \cref{alg:one_step} uses $\hat{\theta}_{\text{plug-in}}$, which averages over datapoints \emph{simulated} from $p_{\hat{\phi}}$. 	
\end{remark}

% \section{MC-EIF Implementation Details} \label{A:mceif_imp_details}

% \TD{Eli - can you fill in this section. Basically, explain how we use hessian vector products and conjugate gradients. You could also explain the general version we have that uses nested Monte Carlo.}

\section{Brief Survey of Recent Specialized Estimators}
\label{A:Survey}
In this section we provide a brief survey of only the past 3 years of work on specialized efficient estimators for particular model-functional combinations.
For TMLE~\cite{tmle}; cluster-randomized trials~\cite{tmle-cluster-randomized}, continuous time-dependent interventions~\cite{tmle-continuous-time}, mixed experimental and observational data~\cite{tmle-experiment-selection}, mediation analysis with longitudinal data~\cite{tmle-longitudinal-mediation}, subgroup treatment effect estimation~\cite{tmle-subgroup}, survival and competing risks analysis~\cite{tmle-survival-competing-risks}, effect estimation with continuous time-to-event outcomes~\cite{tmle-time-to-event}, and variable importance measures for effect estimates~\cite{tmle-variable-importance}. For double/debiased machine learning~\cite{dml}; difference-in-differences~\cite{double-did}, instrumental variable designs~\cite{double-iv}, and mediation analysis~\cite{double-mediation}.

\section{Code Examples} \label{A:code_ex}

\subsection{Automatically Differentiable Functionals} \label{A:autodiff_functionals}

The implementation of differentiable functional approximations is fairly straightforward when using modern autodifferentation tools. For example, the squared density functional for a mean-zero, univariate normal can be approximated using Monte-Carlo as follows:

\begin{equation}
    \frac{1}{N}\sum_{n=1}^N \mathcal{N}\left(x_n, \sigma^2\right);\ \ x_n \thicksim \mathcal{N}\left(0, \sigma^2\right)
\end{equation}

This can be implemented in pytorch \cite{pytorch}, and thus automatically differentiated with respect to $\sigma$ (called ``scale'' in the code block below) using, for example, \texttt{torch.autograd.grad}.

\begin{lstlisting}[language=Python, caption=Automatically Differentiable Monte Carlo Approximation of Integrated Squared Normal Density]
import torch

def diffable_mc_integ_squared_norm_density(scale: torch.Tensor, num_monte_carlo: int):
    assert scale.requires_grad
    # Sample from the density
    samples = torch.distributions.Normal(0., scale).sample((num_monte_carlo,))
    # Evaluate those samples under the density.
    logprobs = torch.distributions.Normal(0., scale).log_prob(samples)
    # Return the mean pdf value in a numerically stable way.
    return torch.exp(torch.logsumexp(logprobs, dim=0)) / torch.numel(samples)
\end{lstlisting}

\section{So, What's Automatic?} \label{A:automatic_comp}

\newcommand{\functional}{\Psi}
\newcommand{\influencefunction}{\varphi}
\newcommand{\expyxawrt}[1]{\E_{#1}\left[Y \mid X, A = 1\right]}
\newcommand{\obs}{O}
\newcommand{\est}[1]{\tilde{#1}}

\begin{figure}[ht]
	\centering
	\includegraphics[width=0.8\textwidth]{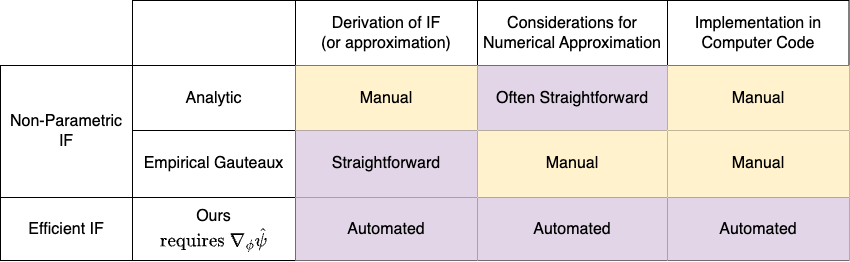}
	\caption{We taxonomize the workflow of robust estimation into three stages: the derivation of an (approximate and/or efficient) influence function, the numerical derivation and analysis required for its computation, and the code required to compute it. For the analytic workflow, the derivation of the IF results in \cref{eq:analytic_if_mpe}. This largely involves terms already required by the original plug-in (\cref{eq:mpe}), but still must be implemented on a case-by-case basis in code. For the ``Empirical Gateaux'' workflow, the first stage requires only the general purpose \cref{eq:emp_gat_if_approx}, but demands case-specific numerical considerations and derivations like the one shown in \cref{eq:emp_gat_mpe_numeric}. In stark contrast, given a differentiable approximation to the functional of interest, $\ourapproach$ ``automates'' each stage through use of an end-to-end, general purpose solution.}
	\label{fig:robustautomationgrid}
\end{figure}

The work required to perform robust estimation can be subdivided into a few key steps. The process begins with a functional of interest, $\Psi$. With this functional in hand, an analyst must first derive the influence function (or an approximation thereof), and consider any nuances in numerically approximating that quantity. Finally, an engineer must implement that approximation as executable code. Different approaches boast varying levels of ``automation'' for each step. We claim that in problems where our conditions hold (as outlined in \cref{sec:theo_eif}), $\ourapproach$ provides end-to-end automation via a general-purpose solution at each stage. Here, we contrast our approach with both the analytic (see e.g. \citet{kennedy-semi-emp-review}) and ``Empirical Gateaux'' workflows \citep{jordanDataDrivenInfluenceFunctions2023, computer-eif}. We will track a workflow's ``products'' at each of the three stages: first, the derivation of the (efficient and/or approximate) influence function; second, a tractable version of the influence function that properly considers its numerical nuances; third, executable code that computes the influence function. Because our approach uses a general purpose formulation for each of these three stages, we call our approach ``automated.''

We assume the workflow starts having identified a functional of interest and having implemented, in code, a plug-in estimator for it. Throughout, we will use the mean-potential outcome (MPE) functional as our working example\footnote{For simplicity in exposition, but without loss of generality in regards to this description of what is and isn't automated, we follow \citet{jordanDataDrivenInfluenceFunctions2023} in assuming that outcome and confounders are continuous real numbers, while the treatment is binary.}:

\begin{equation} \label{eq:mpe}
\functional\left(\p\right) = \E_\p\left[\expyxawrt{\p}\right] = \int \int y\frac{p\left(y,A=1,x\right)}{p\left(A=1, x\right)} p\left(x\right)dydx
\end{equation}

\subsection{Analytic Workflow}

The analytic workflow begins by deriving a closed form influence function — a challenging task even for seasoned experts. This first stage culminates in the following analytic influence function for the MPE \citep{kennedy-semi-emp-review}.

\begin{equation} \label{eq:analytic_if_mpe}
    \influencefunction\left(\obs;\p\right) = \frac{\mathbb{I}\left(A=1\right)}{\p\left(A=1\mid X\right)} \left\{Y - \expyxawrt{\p}\right\} + \expyxawrt{\p} - \functional\left(\p\right)
\end{equation}

For general functionals, the derivation resulting in \cref{eq:analytic_if_mpe} is challenging, even for experts — but given such a derivation, it is often the case that the computation of the quantities composing it can share code and numerical considerations developed to estimate the original, plug-in functional (e.g. \cref{eq:mpe}). Indeed, the influence function of the MPE (\cref{eq:analytic_if_mpe}) involves only terms that an analyst has already considered and implemented for the plug-in (\cref{eq:mpe}). For this reason, we say that in the ``analytic'' workflow, most of the labor must be allocated to deriving the influence function — tractable, well behaved computation of that influence function tends to involve straightforward extensions of tooling and analysis that already exists for the plug-in.

The last stage is the implementation of that tooling in computer code, which will always require some work on a case-by-case basis.

\subsection{Empirical Gateaux Workflow}

The workflow presented by \citet{jordanDataDrivenInfluenceFunctions2023} and \citet{computer-eif} significantly reduces the resources required in the first stage — the derivation of the influence function — by providing a general purpose, finite-difference approximation to the influence function (\cref{eq:emp_gat_if_approx}). $\est{\p}_{\epsilon,\lambda}$, here, represents a perturbation of the estimated distribution $\est{\p}$ of size $\epsilon$ in the direction of a $\lambda$-smooth kernel centered at observation $\obs$.

\begin{equation} \label{eq:emp_gat_if_approx}
    \est{\influencefunction}\left(\obs;\est{\p}\right) = \frac{1}{\epsilon} \left(\functional\left(\est{\p}^\obs_{\epsilon,\lambda}\right) - \functional\left(\est{\p}\right)\right)
\end{equation}

At first glance, it seems that computing this term would follow easily given a general purpose framework for the perturbation of $\est{\p}$, and then applying the plug-in functional to that perturbed density. Unfortunately, computing $\functional\left(\est{\p}^\obs_{\epsilon,\lambda}\right)$ presents a number of numerical challenges in practice. As exhibited in \cref{fig:emp_gat_v_mc_eif} (which echoes figure 1 in the work by \citet{computer-eif}), selecting appropriate perturbation parameters $\epsilon$ and $\lambda$ a priori is challenging, and a battle-tested framework for doing so has not yet been developed. Further, in high dimensions (where $\ourapproach$ excels), the required $\epsilon$ can be so small as to quickly overrun floating point accuracy on modern computers when even $D \approx 10$ \cite{computer-eif}. Indeed, \citet{jordanDataDrivenInfluenceFunctions2023} have explicitly left thorough numerical analysis of this approach to further work. In footnote 7, they anecdotally report that quadrature methods were overly sensitive in evaluating perturbed densities in the MPE functional, and instead present a Monte Carlo approach tailored to the task. Unfortunately, neither the numeric considerations or code-implementations of the plug-in estimator easily translate when computing the plug-in with respect to the \textit{perturbed} data distribution. Below, we show their numeric approximation\footnote{This is their equation 73 in appendix E.3.} of \cref{eq:emp_gat_if_approx} for the MPE, where observation $o$ comprises $(x, a, y)$, the perturbation kernel $K$ has bandwidth $\lambda$, and they use $N$ monte-carlo samples from a uniform kernel over confounder $x$.

\begin{align} \label{eq:emp_gat_mpe_numeric}
\begin{split}
    \est{\influencefunction}_{\lambda, \epsilon}\left(o\right) &= \frac{1}{N}\sum_k \left(\frac{(1-\epsilon)\left( \sum_{j:A_j=1} K(X_j - \tilde{x}_k) Y_j \right) P(A=1) + \epsilon y_i \mathbb{I}\left[a_i=1\right] \cdot 1}{(1-\epsilon)p(A=1, \tilde{x}_k) + \epsilon \mathbb{I}\left[a_i=1\right]} \right) \\ &+ (1 - \epsilon) \frac{1}{N} \sum_k \frac{\tilde{p}(\tilde{x}_k)}{\tilde{p}_\epsilon (A=1, \tilde{x}_k)} \mathbb{I}\left[a_i=1\right] \left\{ y_i - \mathbb{E}_{\tilde{P}} [Y | A=1, \tilde{x}_k] \right\}
\end{split}
\end{align}

Indeed, just like finite differencing can simplify multivariate calculus, but introduce numeric challenges, the empirical-gateaux approach makes variational calculus easier, but introduces numeric challenges. In sum, we consider the second, ``numerical,'' stage of this workflow to be both labor and expertise intensive, even when the curse of dimensionality does not render it moot.

Like in the analytic workflow, the ``coding'' stage of course requires case-by-case implementations. Moreover, added numerical challenges here introduce significant nuance in implementation that isn't present in the analytic case.

\subsection{Our Workflow}

In stark contrast, our workflow exploits general solutions in all three stages for the ``price'' of differentiability of a parametric plug-in estimator. When our general conditions are met (as outlined in \cref{sec:theo_eif}), \cref{eq:eif} provides the general purpose solution to the first stage of deriving an (approximate, efficient) influence function, and the second stage is achieved with the Monte-Carlo approximation in \cref{eq:mceif}.

The third stage is met in software implementing this general purpose solution that operates on functional implementations using one of many auto-differentiation tools now ubiquitous in machine learning (see \cref{A:autodiff_functionals} for a simple example). As discussed at the end of \cref{sec:monte_eif}, this sometimes requires the ability to exploit methods like the reparameterization trick for the functional of interest. In many cases, modern automatic differentiation software makes this trivial (as shown in \cref{A:autodiff_functionals}). For some functionals, however, like those involving inner optimizations, this may be more challenging.

This general purpose approximation underpins the end-to-end automation of $\ourapproach$, and is to the best of our knowledge the only such general purpose approximation for efficient influence functions.

\section{An EIF cookbook} \label{A:eif_cookbook}

In this section we describe some forward-looking generalizations and applications of MC-EIF.

\textbf{Multi-argument functionals} Many important quantities in statistics and machine learning, like the mutual information $\mathbb{I}[X; Y]$ or KL-divergence $\mathbb{KL}[P; Q]$ are functionals of more than one probability distribution. As shown in \cite{kandasamy2015nonparametric}, we can define partial EIFs analogous to partial derivatives for these quantities (which can then be plugged into the efficient estimators of \cref{sec:programmable-DRI}) by treating all but one argument as part of the functional and computing the ordinary EIF with respect to that argument.

\textbf{Higher-order EIFs} Although all of the efficient estimators of \cref{sec:programmable-DRI} are derived from the first-order EIF, there are some circumstances where incorporating higher-order EIFs can be shown to be theoretically necessary for achieving certain statistical properties \cite{robins2008higher,balakrishnan2023fundamental}. Just as ordinary higher-order derivatives are computed by recursively applying a first-order derivative operator to its output, higher-order EIFs can be computed by recursively applying a first-order EIF operator to its own output \cite{van2014higher}, a property straightforwardly inherited by MC-EIF.

\textbf{Models with latent variables} Thus far, we have assumed that we can exactly simulate from model predictive distributions $x \sim p_\phi$ and compute log-densities $\log p_\phi (x)$, score functions $\nabla_\phi \log p_\phi (x)$ and Hessian-vector products. However, our MC-EIF estimator can be extended straightforwardly to models with latent variables and intractable densities and score functions by using a nested Monte Carlo procedure \cite{rainforth2018nesting,syed2023optimal} to approximate the prior predictive or posterior predictive distributions and plugging the resulting stochastic estimates into the vanilla MC-EIF framework. We expect our theoretical results to extend to this case provided the approximation error can be made small relative to the Monte Carlo error in estimating the Fisher matrix from a finite set of model Monte Carlo samples.

\textbf{Infinite-dimensional models and targets} Semiparametric statistics is by definition fundamentally concerned with models that contain infinite-dimensional (i.e. function-valued) components. There is also intense interest in deriving efficient, doubly robust estimators for infinite-dimensional target functionals like the conditional average treatment effect (CATE) in causal inference. Fortunately, in many of these settings the infinite-dimensional quantities can be reduced to finite ones (and ultimately must be to be representable on a digital computer) to which MC-EIF may be applied in a straightforward way. For example, a Gaussian process is fully characterized by the latent function's values on a finite set of test points; computing the EIF for a functional of the GP reduces to computing the EIF of the finite-dimensional joint distribution on function values at the test points, which is straightforward to estimate with MC-EIF. Similarly, in the case of the CATE, we are ultimately interested in the values of the CATE function on a finite set of test inputs, reducing the problem to ordinary MC-EIF for a finite-dimensional target functional.

\section{Additional Experiment Details} \label{A:experiment_details}

\subsection{Model Details}

In \cref{sec:experiments}, we consider the following model with confounders $c$, treatment $t$, and response $y$:
\begin{equation} \label{eq:causal_glm_model}
	\begin{split}
		& \mu_0 \sim  \mathcal{N}(0, 1), \quad \text{(intercept)} \\
		& \xi \sim   \mathcal{N}\left(0, \frac{1}{\sqrt{F}} I_F\right), \quad \text{(outcome weights)} \\
		& \pi \sim   \mathcal{N}\left(0, \frac{1}{\sqrt{F}} I_F\right), \quad \text{(propensity weights)} \\
		& \tau \sim   \mathcal{N}(0, 1), \quad \text{(treatment weight)} \\
		& c_n \sim  \mathcal{N}(0, I_D), \quad \text{(confounders)} \\
		& t_n \mid c_n, \pi \sim \text{Bernoulli}(\text{logits} = \pi^T c_n), \quad \text{(treatment assignment)} \\
		& y_n \sim  \mathcal{N}(\tau t_n + \xi^Tc_n + \mu_0 , 1), \quad \text{(response)} \\
	\end{split}
\end{equation}
where $F \in \mathbb{N}$ denotes the number of confounders. In this example, $x = (c, t, y) \in \R^{D}$, where $D = F + 2$ and $\phi = (\mu_0, \xi, \pi, \tau) \in \R^{2F + 2}$. To obtain a point estimate in \cref{sec:experiments}, we take the maximum a posteriori estimate. In \cref{sec:experiments}, we vary the model dimension $p$ by varying $F$ since $p = 2F + 2$. 

\subsection{Additional Figures}

Here, we provide experimental results that provide interesting insight, but do not directly support the key claims of our paper.

\begin{figure}[ht]
	\centering
	\includegraphics[width=0.75\textwidth]{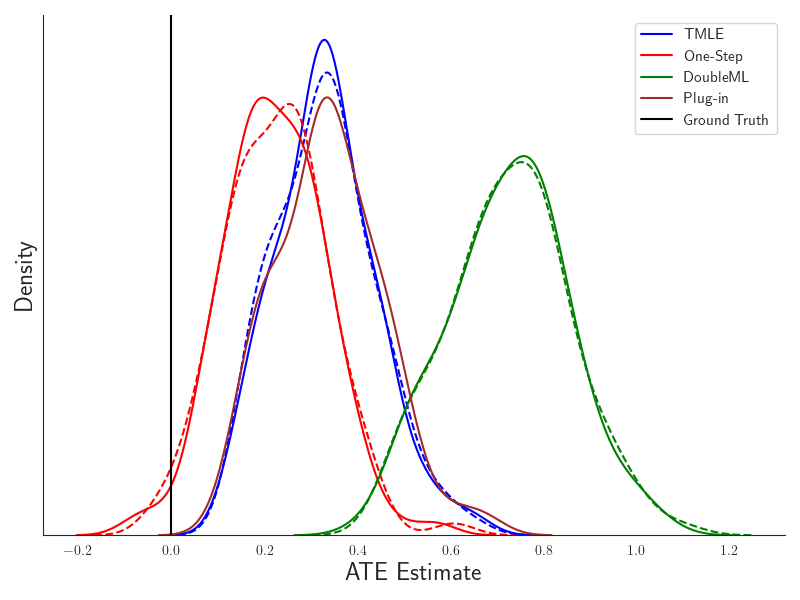}
	\caption{Comparison of plug-in estimator and efficient estimators using MC-EIF and analytic EIF for estimating ATE. The true ATE is 0. Closer to zero the better. The distribution is over 100 simulated datasets. Dashed lines represent the estimates using the analytic EIF, and the solid lines represent using MC-EIF (when applicable).} \label{fig:estimator_comp}
\end{figure}

\begin{figure}[ht]
	\centering
	\includegraphics[width=0.75\textwidth]{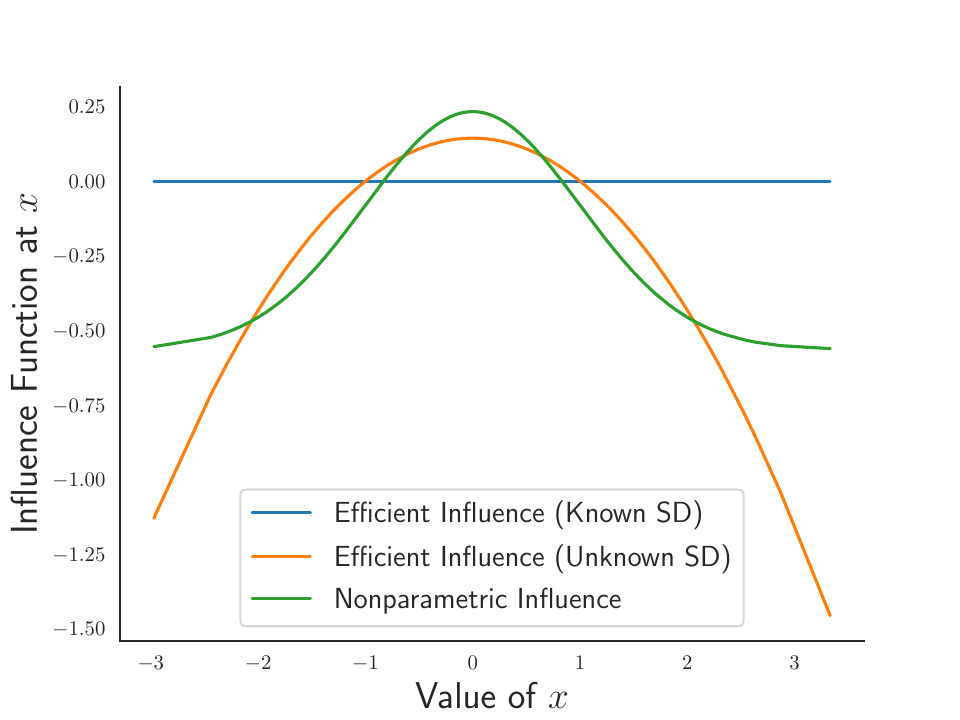}
	\caption{Nonparametric and efficient influence functions for expected density.}
	\label{fig:analytic_inf_density}
\end{figure}
\begin{figure}[ht]
\centering
\includegraphics[width=.75\linewidth]{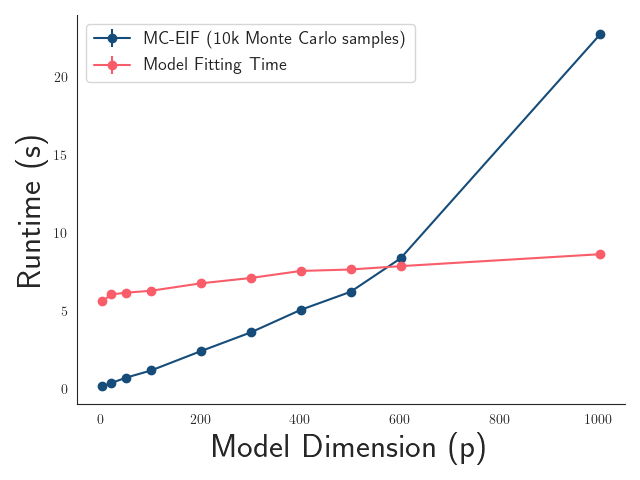}
  \caption{Runtime of fitting point estimate and computing MC-EIF as a function of model size.} \label{fig:dim_sens}
\end{figure}

\begin{figure}[ht]
	\centering
	\includegraphics[width=.75\linewidth]{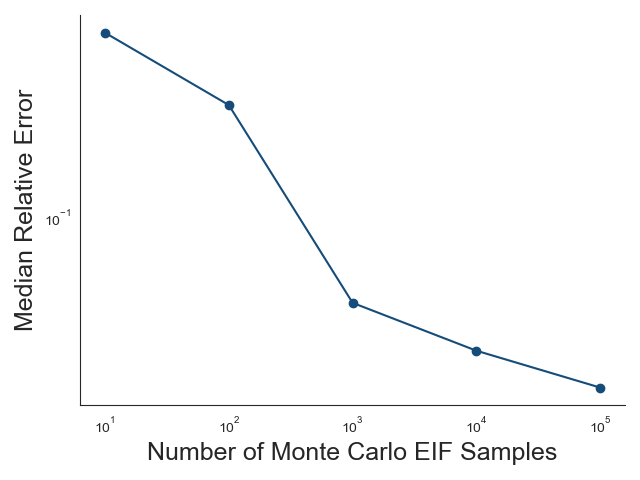}
	\caption{Median relative error between MC-EIF and true efficient influence function for unknown variance model and expected density functional. Median absolute error computed by randomly sampling points to evaluate EIF, computing the relative error at each point, and then taking the median.} \label{fig:monte_carlo_increase_density}
\end{figure}

\end{document}